\documentclass[11pt]{article}

\usepackage[dvipsnames]{xcolor}
\usepackage{hyperref}
\usepackage{amsmath,amssymb,bbm,amsthm}
\usepackage{fullpage,wrapfig}
\usepackage{thm-restate,color,xcolor,xspace,nicefrac}
\usepackage{mathtools,xspace}
\usepackage[T1]{fontenc}
\usepackage{float}

\usepackage{graphicx}
\usepackage{xcolor}
\usepackage{enumitem}
\usepackage{algorithm}
\usepackage[noend]{algpseudocode}
\usepackage{comment}
\usepackage[normalem]{ulem}

\usepackage[capitalise,nameinlink]{cleveref}

\usepackage[format=plain,
            labelfont=it,
            textfont=it]{caption}
            
\newcommand\myshade{85}
\colorlet{mylinkcolor}{violet}
\colorlet{mycitecolor}{YellowOrange}
\colorlet{myurlcolor}{Aquamarine}
\hypersetup{
  linkcolor  = mylinkcolor!\myshade!black,
  citecolor  = mycitecolor!\myshade!black,
  urlcolor   = myurlcolor!\myshade!black,
  colorlinks = true,
}

\makeatletter
\allowdisplaybreaks
\makeatother

\usepackage{setspace}
\usepackage[compact]{titlesec}
\allowdisplaybreaks

\newcommand{\sse}{\subseteq}

\newcommand{\BE}{\begin{enumerate}}
\newcommand{\EE}{\end{enumerate}}

\newcommand{\BI}{\begin{itemize}}
\newcommand{\EI}{\end{itemize}}

\newcommand{\R}{\mathbb R}

\newcommand{\pr}{\mathbb P}

\newcommand{\eps}{\varepsilon}

\newcommand{\E}{\mathbb E}

\DeclareMathOperator{\Ber}{Ber}

\newcommand{\cD}{\mathcal{D}}
\newcommand{\cE}{\mathcal{E}}

\newcommand{\cM}{\mathcal{M}}
\newcommand{\cP}{\mathcal{P}}

\newcommand{\cT}{\mathcal{T}}

\newcommand{\lmt}{\left[\begin{matrix}}
\newcommand{\rmt}{\end{matrix}\right]}

\newtheorem{theorem}{Theorem}[section]
\newtheorem{lemma}[theorem]{Lemma}
\newtheorem{definition}[theorem]{Definition}
\newtheorem{corollary}[theorem]{Corollary}
\newtheorem{observation}[theorem]{Observation}
\newtheorem{claim}[theorem]{Claim}
\newtheorem{subclaim}{Subclaim}
\newtheorem{fact}[theorem]{Fact}

\newtheorem{assumption}[theorem]{Assumption}
\newcommand{\BT}{\begin{theorem}}
\newcommand{\ET}{\end{theorem}}
\newcommand{\BL}{\begin{lemma}}
\newcommand{\EL}{\end{lemma}}
\newcommand{\BD}{\begin{definition}}
\newcommand{\ED}{\end{definition}}
\newcommand{\BC}{\begin{corollary}}
\newcommand{\EC}{\end{corollary}}
\newcommand{\BO}{\begin{observation}}
\newcommand{\EO}{\end{observation}}
\newcommand{\BCL}{\begin{claim}}
\newcommand{\ECL}{\end{claim}}
\newcommand{\BSCL}{\begin{subclaim}}
\newcommand{\ESCL}{\end{subclaim}}
\newcommand{\BF}{\begin{fact}}
\newcommand{\EF}{\end{fact}}
\newcommand{\BA}{\begin{assumption}}
\newcommand{\EA}{\end{assumption}}
\newcommand{\BP}{\begin{proof}}
\newcommand{\EP}{\end{proof}}
\newcommand{\BPS}{\begin{proof}[Proof (Sketch)]}
\newcommand{\EPS}{\end{proof}}
\Crefname{observation}{Observation}{Observations}
\Crefname{claim}{Claim}{Claims}
\Crefname{subclaim}{Subclaim}{Subclaims}
\Crefname{fact}{Fact}{Facts}
\Crefname{assumption}{Assumption}{Assumptions}

\newcommand{\ignore}[1]{}

\renewcommand{\emptyset}{\varnothing}

\newcommand{\odd}{\text{ odd}}

\newcommand{\reduced}{\text{ reduced}}
\newcommand{\haff}{\nicefrac12}

\newcommand{\nf}{\nicefrac}
\renewcommand{\wp}{w.p.\ }
\newcommand{\yh}{\widehat{y}}
\newcommand{\Hh}{\widehat{H}}

\newcommand{\initOneLiners}{%
    \setlength{\itemsep}{0pt}
    \setlength{\parsep }{0pt}
    \setlength{\topsep }{0pt}
}
\newenvironment{OneLiners}[1][\ensuremath{\bullet}]
    {\begin{list}
        {#1}
        {\initOneLiners}}
    {\end{list}}

\newcommand{\Ktsp}{K_{TSP}}
\newcommand{\Kpm}{K_{PM}}

\newcommand{\Kst}{K_{spT}}
\newcommand{\Kjoin}{K_{join}}

\newcommand{\lAngle}{\langle\!\langle}
\newcommand{\rAngle}{\rangle\!\rangle}
\newcommand{\dip}[1]{{\lAngle #1 \rAngle}}
\newcommand{\ones}{\mathbf{1}}

\newcommand{\red}{\text{ reduced}}
\newcommand{\contract}{G\lAngle S \rAngle}
\newcommand{\ex}{\mathbb{E}}
\newcommand{\rr}{\mathbb{R}}

\newcommand{\MI}{\textsc{MatInt}\xspace}
\newcommand{\ME}{\textsc{MaxEnt}\xspace}

\newif\ifArXivVersion

\ArXivVersiontrue

\begin{document}

\title{Matroid-Based TSP Rounding for Half-Integral Solutions}

\author{Anupam Gupta\thanks{Carnegie Mellon University, Pittsburgh PA
    15217.} \and Euiwoong Lee\thanks{University of Michigan, Ann Arbor
    MI 48109. Most of this work was done when the author was a postdoc at NYU and supported in part by the Simons Collaboration on Algorithms and Geometry. } \and Jason Li\thanks{Simons Institute for the Theory of Computing and UC Berkeley, Berkeley CA 94720.}\and Marcin Mucha\thanks{University of Warsaw, Warsaw, Poland.}
 \and Heather Newman$^*$ \and Sherry Sarkar$^*$}

\maketitle

\begin{abstract}
  We show how to round any half-integral solution to the
  subtour-elimination relaxation for the TSP, while losing a
  less-than-1.5 factor. Such a rounding algorithm was recently given
  by Karlin, Klein, and Oveis Gharan based on sampling from
  max-entropy distributions. We build on an approach of Haddadan and
  Newman to show how sampling from the matroid intersection polytope,
  and a new use of max-entropy sampling, can give better guarantees.
\end{abstract}

\section{Introduction}
\label{sec:introduction}

The (symmetric) traveling salesman problem asks: given an graph
$G = (V,E)$ with edge-lengths $c_e \geq 0$, find the shortest tour
that visits all vertices at least once. The Christofides-Serdyukov algorithm~\cite{Chr,Ser} gives a
$\nf32$-approx\-imation to this APX-hard problem; this was recently
improved to a $(\nf32-\eps)$-approx\-imation by the breakthrough work of
Karlin, Klein, and Oveis Gharan, where $\eps > 0$~\cite{KKO20}. A
related question is: \emph{what is the integrality gap of the
  subtour-elimination polytope relaxation for the TSP}? Wolsey had
adapted the Christofides-Serdyukov analysis to show an upper bound of
$\nf32$~\cite{Wol80} (also~\cite{SW90}), and there exists a lower bound of
$\nf43$. Building on their above-mentioned work, Karlin, Klein, and
Oveis Gharan gave an integrality gap of $1.5 - \eps'$ for another
small constant $\eps' > 0$~\cite{KKO21}, thereby making the first
progress towards the conjectured optimal value of $\nf43$ in nearly half a
century.

Both these recent results are based on a randomized version of the
Christofides-Serdyukov algorithm proposed by Oveis Gharan,
Saberi, and Singh~\cite{OSS11}. This algorithm first samples a
spanning tree (plus perhaps one edge) from the \emph{max-entropy
  distribution} with marginals matching the LP solution, and adds an
$O$-join on the odd-degree vertices $O$ in it, thereby getting an
Eulerian spanning subgraph. Since the first step has expected cost
equal to that of the LP solution, these works then bound the cost of
this $O$-join by strictly less than half the optimal value, or the LP
value. The proof uses a cactus-like decomposition of the min-cuts of
the graph with respect to the values $x_e$, like in~\cite{OSS11}.

Given the $\nf32$ barrier has been broken, we can ask: what other
techniques can be effective here? How can we make further progress?
These questions are interesting even for cases where the LP has
additional structure. The half-integral cases are particularly
interesting due to the Schalekamp, Williamson, and van Zuylen
conjecture, which says that the integrality gap is achieved on
instances where the LP has optimal half-integral
solutions~\cite{SWZ14}. The team of Karlin, Klein, and Oveis Gharan
first used their max-entropy approach to get an integrality gap of
$1.4998$ for half-integral LP solutions~\cite{KKO19}, before they
moved on to the general case in \cite{KKO20} and obtained an
integrality gap of $1.5-\varepsilon$; the latter improvement
is considerably smaller than in the half-integral case.  It is natural
to ask: can we do better for half-integral
instances? %

In this paper, we answer this question affirmatively. We show how to
get tours of expected cost at most 
$1.49842$ times the linear program
value using an algorithm based on matroid intersection. Moreover, some
of these ideas can be used to strengthen the max-entropy sampling approach in the
half-integral case. The matroid intersection approach and the strengthened max-entropy approach each yield improvements over the bound in \cite{KKO19}. Combining the techniques gives our final quantitative improvement:

\begin{theorem}
  \label{thm:main}
  Let $x$ be a half-integral solution to the subtour elimination
  polytope with cost $c(x)$. There is a randomized algorithm that
  rounds $x$ to an integral solution whose cost is at most
  ${(1.5-\varepsilon)\cdot c(x)}$, where 
  $\varepsilon = 0.00158$. 
\end{theorem}

We view our work as showing a proof-of-concept of
the efficacy of combinatorial techniques (matroid
intersection, and flow-based charging arguments) in getting an
improvement for the half-integral case. We hope that these techniques,
ideally combined with max-entropy sampling techniques, can give
further progress on this central problem.

\paragraph{Our Techniques}

The algorithm is again in the Christofides-Serdyukov framework. It is
easiest to explain for the case where the graph (a)~has an even number
of vertices, and (b)~has no (non-trivial) proper min-cuts with respect
to the LP solution values $x_e$---specifically, the only sets for
which $x(\partial S) = 2$ correspond to the singleton cuts. Here, our
goal is that each edge is ``even'' with some probability: i.e., both of
its endpoints have even degree with probability $p > 0$. In this case
we use an idea due to Haddadan and Newman~\cite{HN19}: we \emph{shift}
and get a $\{\nf13,1\}$-valued solution $y$ to the subtour elimination
polytope $\Ktsp$. Specifically, we find a random perfect matching $M$
in the support of $x$, and set $y_e = 1$ for $e\in M$, and $\nf13$
otherwise, thereby ensuring $\E[y] = x$. To pick a random tree from
this shifted distribution $y$, we do one of the following:
\begin{enumerate}
\item We pick a random ``independent'' set $M'$ of matching edges (so
  that no edge in $E$ is incident to two edges of $M'$). For each
  $e' \in M'$, we place partition matroid constraints enforcing that
  exactly one edge is picked at each endpoint---which, along with $e'$
  itself, gives degree $2$ and thereby makes the edge even as
  desired. Finding spanning trees subject to another matroid
  constraint can be implemented using matroid intersection.
\item Or, instead we sample a random spanning tree from the
  max-entropy distribution, with marginals being the shifted value
  $y$. (In contrast, \cite{KKO19} sample trees from $x$ itself; our
  shifting allows us to get stronger notions of evenness than they do:
  e.g., we can show that every edge is ``even-at-last'' with constant
  probability, as opposed to having at least one even-at-last edge in
  each tight cut with some probabiltiy.) %
\end{enumerate}
(Our algorithm randomizes between the two samplers to achieve the best
guarantees.) For the $O$-join step, it suffices to give fractional
values $z_e$ to edges so that for every odd cut in $T$, the $z$-mass
leaving the cut is at least 1. In the special case we consider, each
edge only participates in two min-cuts---those corresponding to its
two endpoints. %
So set $z_e = \nf{x_e}3$ if $e$ is even, and $\nf{x_e}2$ if not; the
only cuts with $z(\partial S) < 1$ are minimum cuts, and these cuts
will not show up as $O$-join constraints, due to
evenness.
For
this setting, if an edge is even with probability $p$, we get a
$(\nf32 -
\nf{p}6)$-approximation! %

It remains to get rid of the two simplifying assumptions. To sample
trees when $|V|$ is odd (an open question from~\cite{HN19}), we add a
new vertex to fix the parity, and perform local surgery on the
solution to get a new TSP solution and reduce to the even case. The
challenge here is to show that the losses incurred are small, and
hence each edge is still even with constant probability.

Finally, what if there are proper tight sets $S$, i.e., where
$x(\partial S) = 2$? We use the cactus decomposition of a graph (also used
in \cite{OSS11,KKO19}) to sample spanning trees from pieces of $G$
with no proper min-cuts, and stitch these trees together. These pieces
are formed by contracting sets of vertices in $G$, and have a
hierachical structure. Moreover, each such piece is either of the form
above (a graph with no proper min-cuts) for which we have already seen
samplers, or else it is a double-edged cycle (which is easily sampled
from). Since each edge may now lie in many min-cuts, we no longer just
want an edge to have both endpoints be even. Instead, we use an idea
from~\cite{KKO19} that uses the hierarchical structure on the pieces
considered above. Every edge of the graph is ``settled'' at exactly
one of these pieces, and we ask for %
both of its endpoints to have even degree \emph{in the piece at which
  it is settled}. The $z_e$ value of such an edge may be lowered in
the $O$-join without affecting constraints corresponding to cuts
\emph{in the piece at which it is settled}.

Since cuts at other
levels of the hierarchy may now be deficient because of the lower values
of $z_e$, %
we may need to increase the $z_f$ values for other ``lower'' edges $f$
to satisfy these deficient cuts. This last part requires a charging
argument, showing that each edge $e$ has $z_e$ that is strictly
smaller than $\nf{x_e}2$ in expectation. For our samplers, 
the na\"{\i}ve approach of distributing charge uniformly as in ~\cite{KKO19} does not work,
so we instead
formulate this charging as a flow problem. 

\section{Notation and Preliminaries}
\label{sec:defn}

Given a multigraph $G = (V,E)$, and a set $S \sse V$, let $\partial S$
denote the cut consisting of the edges connecting $S$ to
$V\setminus S$; $S$ and $\bar{S} \coloneqq V \setminus S$ are called \emph{shores} of the cut. (For a singleton set $\{v\}$, we write $\partial v$
instead of $\partial\{v\}$.) A subset $S \sse V$ is \emph{proper} if
$1 < |S| < |V|-1$; a cut $\partial S$ is called \emph{proper} if the
set $S$ is a proper subset.  A set $S$ is \emph{tight} if $|\partial S|$
equals the size of the minimum edge-cut in $G$. Two sets $S$ and $S'$
are \emph{crossing} if $S \cap S'$, $S \setminus S'$, $S' \setminus
S$, and $V \setminus (S \cup S')$ are all non-empty.

Define the \emph{subtour elimination polytope}
$\Ktsp(G) \sse \R^{|E|}$:
\begin{align}
  x(\partial v) &= 2 \qquad\qquad \forall v \in V \tag{LP-TSP} \label{eq:HK}\\
  x(\partial S) &\geq 2 \qquad\qquad\forall \;\text{proper $S$} \notag\\
  x \geq 0. \notag
\end{align}

Let $x$ be half-integral and feasible for~(\ref{eq:HK}).  W.l.o.g.\ we
can focus on solutions with $x_e = \haff$ for each $e \in E$, doubling
edges if necessary. The support graph $G$ is then a $4$-regular
$4$-edge-connected (henceforth 4EC) multigraph.

The \emph{spanning tree polytope} $\Kst(G) \sse \R^{|E(G)|}$ for a
multigraph $G$ is:
\begin{align}
  x(E(S)) &\leq |S| - 1 \qquad\qquad\forall S \sse V(G) \tag{LP-spT} \label{eq:st-poly}\\
  x(E(V(G))) &= |V(G)|-1 \notag \\
  x &\geq 0. \notag
\end{align}
where $E(S)$ is the set of edges with both endpoints in $S$. This is
the graphic matroid polytope, and the convex hull of the spanning trees. 

\begin{definition}[$r$-Trees]
Given a multigraph $G$ and ``root'' vertex $r \in V(G)$, an
\emph{$r$-tree} $T$ is a connected subgraph with $n$ edges: the vertex $r$
has degree exactly $2$, and the subgraph restricted to the other
vertices $V(G) \setminus \{r\}$ is a spanning tree on them. This is a
matroid, though we do not use this fact.
\end{definition}

The (integral) \emph{perfect matching polytope} $\Kpm \sse \R^{|E(G)|}$ is defined as follows: 
\begin{align}
  x(\partial v) &= 1 \qquad\qquad \forall v \in V \tag{LP-PM} \label{eq:PM}\\
  x(\partial S) &\geq 1 \qquad\qquad\forall \text{$S$ with $|S|$ odd} \notag\\
  x \geq 0. \notag
\end{align}

The (integral) \emph{$O$-join dominator polytope} $K_{join}(G,O)$ is defined as follows. Let $O \subseteq V$, $|O|$ even. 
\begin{align*}
    z(\partial(S)) &\geq 1 \qquad\qquad \forall S \subseteq V, |S \cap O| \text{ odd} \\
    z &\geq 0
\end{align*}

\begin{fact}
  \label{fact:tsp-implies-onetree}
  For any solution $x \in \Ktsp(G)$, it holds that $x|_{E(V(G)\setminus\{r\})} \in \Kst(G[V(G) \setminus \{r\}])$, $\nf{x}2 \in
  \Kpm(G)$ (when $|V(G)|$ is even), and $\nf{x}{2} \in K_{join}(G,O)$ for $O \subseteq V(G)$, $|O|$ even. 
\end{fact}

\begin{lemma}
  \label{lem:sampler}
  Consider a sub-partition $\cP = \{P_1, P_2, \ldots, P_t\}$ of the
  edge set of $G$. Let $x$ be a fractional solution
  to~(\ref{eq:st-poly}) that satisfies $x(P_i) \leq 1$ for all
  $i \in [t]$. Then we can efficiently sample from a probability
  distribution $\cD$ over spanning trees which contain at most one edge
  from each of the parts $P_i$, such that
  $ \Pr_{T \gets \cD}[e \in T] = x_e$. 
\end{lemma}
\begin{proof}
  This follows from the
  integrality of the matroid intersection polytope.
\end{proof}

\subsection{The Max-Entropy Distribution over Spanning Trees}

\begin{definition}[Strong Rayleigh Distributions, \cite{BBL09}]
  Let $\mu$ be a probability distribution on spanning trees. The
  generating polynomial for $\mu$ is
  $p(z) = \sum_{T \in \mathcal{T}} \mathbb{P}(e \in T) \prod_{e \in
    T}z_e$.  We say $\mu$ is \emph{strongly Rayleigh (SR)} if $p$ is
  a real stable polynomial (i.e., $p(z) \neq 0$ if $z$ lies in the
  upper half plane).
\end{definition}
 
A distribution $\mu$ over spanning trees is called $\lambda\textit{-uniform}$ or $\textit{weighted uniform}$ %
if there exist non-negative weights $\lambda: E \rightarrow \rr$ such that 
$\mathbb{P}(T) \propto \prod_{e \in T} \lambda(e)$. Borcea et al. \cite{BBL09} showed that $\lambda$-uniform spanning tree distributions are SR. 
\begin{theorem}[\cite{OSS11}]
There exists a $\lambda$-uniform distribution $\mu$ over spanning trees such that given $z$ in the spanning tree polytope of $G = (V, E)$,
\[ \sum_{T \ni e} \mathbb{P}_\mu(T) = z_e. \]
The \emph{max-entropy distribution}, which is the distribution $\mu$ on spanning trees maximizing the entropy of $\mu$ subject to preserving the marginals $z_e$ (that is, $\pr_{\mu}(e \in T)=z_e$), is one such distribution. 
\end{theorem}

Asadpour et al \cite{AGMGS17} showed that we can find weights $\Tilde{\lambda}: E \rightarrow \rr^+$ which approximately respect the marginals given by a vector $z$ in the spanning tree polytope in polynomial time. 
\begin{theorem}[\cite{AGMGS17}]\label{approximate-lambda-uniform}
Given $z$ in the spanning tree polytope of $G = (V, E)$ and some $\varepsilon > 0$, values $\Tilde{\lambda}_e$ for $e \in E$ can be found such that, for all $e \in E$ , the $\Tilde{\lambda}$-uniform distribution $\mu$ satisfies
\[ \sum_{T \ni e} \mathbb{P}_{\mu}(T) \leq (1 + \varepsilon) z_e. \]
The running time is polynomial in $|V|, \log \nicefrac{1}{\min_e z_e}$ and $\log \nicefrac{1}{\varepsilon}$.
\end{theorem}

\begin{fact}[\cite{BBL09}]
  Let $\mu$ be an SR distribution. Let $F \subset E$ be a subset of
  edges. Then the projection of $\mu$ onto $F$, i.e.,
  $\mu_{|F}(A) = \sum_{S: S \cap F = A} \mu(S)$, is also SR.
  Moreover, for any edge $e \in E$, conditioning on $e \in T$ or on $e
  \not \in T$ preserves the SR property.   
\end{fact}

\begin{theorem}[Negative Correlation,~\cite{OSS11}]\label{thm:SRfacts}
  Let $\mu$ be an SR distribution on spanning trees. 

  \begin{enumerate}
  \item Let $S$ be a set of edges and $X_S = |S \cap T|$, where $T \sim \mu$. Then, 
    $X_S \sim \sum_{i = 1}^{|S|} Y_i$,
    where the $Y_i$ are independent Bernoulli random variables with success probabilities $p_i$ and $\sum_i p_i = \mathbb{E}[X_S]$ . 
  \item For any set of edges $S$ and $e \not \in S$,   
    \begin{OneLiners}
    \item[(i)] $\mathbb{E}_{\mu}[X_S] \leq \mathbb{E}_{\mu}[X_S \mid
      X_e = 0] \leq \mathbb{E}_{\mu}[X_S] + \mathbb{P}_{\mu}(e \in
      T)$, and
    \item[(ii)] $\mathbb{E}_{\mu}[X_S]-1+\mathbb{P}_{\mu}(e \in T) \leq \mathbb{E}_{\mu}[X_S \mid X_e =1] \leq \mathbb{E}_{\mu}[X_S]$.
    \end{OneLiners}
  \end{enumerate} 
\end{theorem}

\begin{theorem}[\cite{Hoe56}, Corollary 2.1]\label{thm:Hoeff}
Let $g: \{ 1, \hdots, m\} \rightarrow \rr$ and let $0 \leq p \leq
m$. Let $B_1, \hdots, B_m$ be Bernoulli random variables with
probabilities $p_1^*, \hdots, p_m^*$ that maximize (or minimize) ${\ex[g(B_1+ \cdots + B_m)]}$
over all possible success probabilities $p_i$ for $B_i$ for which $p_1 + \cdots + p_m = p$. Then $\{p_1^*, \hdots, p_m^*\} \in \{0, x, 1\}$ for some $x \in (0, 1)$. 
\end{theorem}

\section{Samplers} \label{sec:samplers}

In this section, we describe the \ME and \MI
samplers for %
graphs that contain no
proper min-cuts. We give bounds on %
certain correlations between edges that will be used in
\S\ref{sec:good-edges} to prove that every edge is ``even'' with constant probability. The
samplers for the case where the graph has an odd number of vertices
are more technical and are deferred to Appendix~\ref{sec:odd-sampler}.

Suppose the graph
$H = (V,E)$ is 4-regular and 4-edge-connected (4EC), contains at least four
vertices, and has no proper min-cuts.\footnote{This implies that $H$ is a
  simple graph: since parallel edges between $u,v$ means that
  $\partial(\{u,v\})$ is a proper min-cut; we use this simplicity of
  the graph often in the arguments of this section.} This means all
proper cuts have six or more edges. 
We are given a dedicated
\emph{external} vertex $r \in V(H)$; the vertices
$I := V \setminus \{r\}$ that are not external are called
\emph{internal}.
(In future sections, this vertex $r$ will be given by a cut hierarchy.) 
Call the edges in $\partial r$ 
\emph{external edges}; all other edges are \emph{internal}. An
internal vertex is called a \emph{boundary vertex} if it is adjacent
to $r$. An edge is said to be \emph{special} if both of its endpoints
are non-boundary vertices.  

We show two ways to sample a spanning tree on $H[I]$, the graph induced on
the internal vertices, being faithful to the marginals $x_e$, i.e.,
$\mathbb{P}_T(e \in T)= x_e$ for
all $e \in E(H) \setminus \partial r$.
Moreover, we want that for each internal edge, both its endpoints have
even degree in $T$ with constant probability. This property will allow
us to lower the cost of the $O$-join in \S\ref{sec:charging-argument}.
While both samplers will satisfy this property, each will do better in
certain cases. The \MI sampler targets special edges; it allows us to
randomly ``hand-pick'' edges of this form and enforce that both of its
endpoints have degree 2 in the tree. The \ME sampler, on the other
hand, relies on maximizing the randomness of the spanning tree sampled
(subject to being faithful to the marginals); negative correlation
properties allow us to obtain the evenness property, and in
particular, better probabilities than \MI for non-special edges, and a
worse probability for the special edges.

Our samplers 
will depend on the parity of $|V|$: when $|V|$ is even, %
the \MI sampler is the %
one given by~\cite[Theorem~13]{HN19}, which we describe in
\S\ref{sec:even-sampler}. %
They left the case of odd $|V|$ as an open problem, and in Appendix \ref{sec:odd-sampler}
we extend their procedure to the odd case. %

\subsection{Samplers for Even $|V(H)|$}
\label{sec:even-sampler}

Since $H$ is $4$-regular and 4EC and $|V(H)|$ is even, setting a value
of $\nf14 = \nf{x_e}{2}$ on each edge gives a solution to the perfect
matching polytope $K_{PM}(H)$ by \Cref{fact:tsp-implies-onetree}.
  
\begin{enumerate}
\item Sample a perfect matching $M$ such that
  $\mathbb{P}(e \in M) = \nf14 = x_e/2$ for all $e \in E(H)$.
\item (Shift) Define a new fractional solution $y$ (that depends on $M$): set
  $y_e = 1$ for $e \in M$, and $y_e = \nf13$ otherwise.  We have
  $y \in K_{TSP}(H)$ (and hence $y|_I \in \Kst(H[I])$ by
  \Cref{fact:tsp-implies-onetree}): %
  indeed, each vertex has 
  $y(\partial v) = 1 + 3\cdot \nf13 = 2$ because $M$ is a perfect
  matching. Moreover, every proper cut $U$ in $H$ has at least
  six edges, so $y(\partial U) \geq |\partial U| \cdot \nf13 \geq 2$. 
  Furthermore, 
  \begin{gather}
    \E_M[ y_e ] = \nf14\cdot 1 + \nf34\cdot \nf13 = \nf12 = x_e. \label{eq:y-is-x}
  \end{gather}

\item Sample a spanning tree faithful to the marginals $y$, using one of two samplers:

  \begin{enumerate}
  \item \ME Sampler: Sample from the max-entropy distribution on
    spanning trees with marginals $y$.
  \item \MI Sampler:
    \begin{enumerate}
    \item \label{item:pick-Mp-even} Color the edges of $M$ using $7$
      colors such that no edge of $H$ is adjacent to two edges of $M$
      having the same color; e.g., by greedily $7$-coloring the
      $6$-regular graph $H/M$. Let $M'$ be one of these color classes
      picked uniformly at random. Hence, $\mathbb{P}(e \in M') = \nf1{28}$,
      and $\mathbb{P}(\partial v \cap M' \neq \emptyset) = \nf17$.
    \item \label{item:beetles-even} For each edge $e = uv \in M'$, let
      $L_{uv}$ and $R_{uv}$ be the sets of edges incident at $u$ and
      $v$ other than $e$. Note that $|L_{uv}| = |R_{uv}| = 3$.
      Place partition matroid constraints $y(L_{uv}) \leq 1$ and
      $y(R_{uv}) \leq 1$ on each of these sets. Finally, restrict the
      partition constraints to the internal edges of $H$; this means
      some of these constraints are no longer tight for the solution
      $y$.
      
    \end{enumerate}
  \item %
    Given   the sub-matching $M' \sse M$, and the
    partition matroid $\cM$ on the internal edges defined using
    $M'$, %
    use~\Cref{lem:sampler} to sample a tree on
    $H[V\setminus \{r\}]$ (i.e., on the internal vertices and edges of
    $H$) with marginals $y_e$, subject to this partition matroid $\cM$.
  \end{enumerate}
\end{enumerate}
Conditioned on the matching $M$, we have %
$\mathbb{P}(e \in T \mid M) = y_e$; now using~(\ref{eq:y-is-x}), we have $\mathbb{P}(e \in T) =
x_e$ for all $e \in (E \setminus
  \partial r)$.

\begin{figure}\centering
\includegraphics[scale=.3]{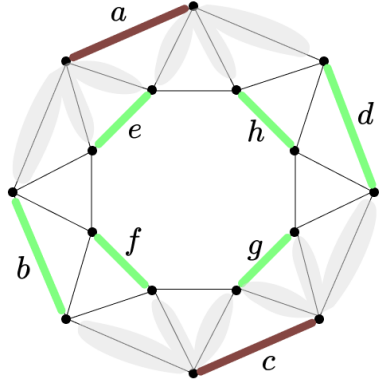}
\caption{The matching $M$ consists of all highlighted edges (both brown
    and green), one possible choice of $M'$ has edges $\{a,c\}$ highlighted in brown, and the
  constraints are placed on the edges adjacent to those in
    $M'$ (marked in gray). }
\label{fig:beetle}
\end{figure}

\subsection{Correlation Properties of Samplers}
\label{sec:corr-props}

Let $T$ be a tree sampled using either the \MI or the \ME sampler. The following claims will be used to prove the evenness property in \S\ref{sec:good-edges}. Each table gives lower bounds on the corresponding probabilities for each sampler. The proofs for $|V(H)|$ odd are in Appendix \ref{sec:odd-sampler}.

\begin{restatable}{lemma}{Oddfg}
  \label{clm:odd1}
  If $f,g$ are internal edges incident to a vertex $v$, then
  \begin{table}[h]
    \centering
    \begin{tabular}{|l|l|l|}
      \hline
      \textbf{Probability Statement}   & \textbf{\MI} & \textbf{\ME}           \\ \hline
      $\mathbb{P}( |T \cap \{f, g\}| = 2)$     & $\nicefrac{1}{9}$    & $\nicefrac{1}{9}$      \\ \hline
      $\mathbb{P}( T \cap \{f, g\} = \{f\})$     & $\nicefrac{1}{9}$    & $\nicefrac{12}{72}$    \\ \hline
    \end{tabular}
  \end{table}
 \end{restatable}

\begin{restatable}{lemma}{Oddefgh}
  \label{clm:odd2}
  If edges $e,f,g,h$ incident to a vertex $v$ are all internal, then
  \begin{table}[H]
    \centering
    \begin{tabular}{|l|l|l|}
      \hline
      \textbf{Probability Statement}  & \textbf{\MI} & \textbf{\ME}            \\ \hline
      $\mathbb{P}( |T \cap \{e,f,g,h\} | = 2 )$   & $\nicefrac{2}{21}$   & $\nicefrac{8}{27}$     \\ \hline
      $\mathbb{P}(|T \cap \{e, f\}| = |T \cap \{g, h\}| = 1)$  & $\nicefrac{4}{63}$   & $\nicefrac{16}{81}$    \\ \hline
    \end{tabular}
  \end{table}
\end{restatable}

\begin{restatable}{lemma}{OddGen}
  \label{clm:odd4}
  For an internal edge $e = uv$: 
  \begin{OneLiners}
  \item[(a)] \label{item:odd4a} if both endpoints are non-boundary vertices, then
    \begin{table}[H]
      \centering
      \begin{tabular}{|l|l|l|}
        \hline
        \textbf{Probability Statement}                                                                                   & \textbf{\MI} & \textbf{\ME}            \\ \hline
        $\mathbb{P}( |\partial_T(u)| = |\partial_T(v)| = 2)$                  & $\nicefrac{1}{36}$   & $\nicefrac{128}{6561}$ \\ \hline
      \end{tabular}
    \end{table}
 
  \item[(b)] \label{item:odd4c} if both $u,v$ are boundary vertices, then
    \begin{table}[H]
      \centering
      \begin{tabular}{|l|l|l|}
        \hline
        \textbf{Probability Statement}                                                                                   & \textbf{\MI} & \textbf{\ME}            \\ \hline
        $\mathbb{P}($exactly one of $u,v$ has odd degree in $T)$ & $\nicefrac{1}{9}$    & $\nicefrac{5}{18}$   \\ \hline
      \end{tabular}
    \end{table}
  \end{OneLiners}
\end{restatable}

\subsubsection{Correlation Properties: $|V(H)|$ even}
\label{sec:corr-prop-even}

We now prove the correlation properties for the even case: the
numerical bounds for the even case are better than those claimed above (which will be dictated by the
proofs of the odd case; see Appendix \ref{sec:odd-sampler}).

\begin{proof}[Proof of \Cref{clm:odd1}, Even Case] \label{3.1evencase}
  To prove $\pr(|T\cap\{f,g\}|=2) \geq \nf19$, we need only knowledge of the marginals and not the specific sampler. If one of $f,g$ lies in $M$ (which happens \wp \nf12), then its $y$-value equals $1$ and it belongs to $T$ w.p. $1$, and the
  other edge is chosen \wp $\nf13$, making the
  unconditional probability $\nf12 \cdot \nf13 =
  \nf16 \geq \nf19$. %
  Similarly, conditioned on $f$ lying in $M$ and hence belonging to
  $T$, edge $g$ is not chosen \wp $1 - y_e = \nf23$, so $\pr(T\cap \{f,g\} = \{f\}) \geq \nf14 \cdot \nf23 = \nf16 \geq \nf19$.
  
  \emph{The \ME claim: } It remains to show that $\pr(T\cap \{f,g\}=\{f\})\geq \nf{12}{72}$. Conditioned on $f \in M$, we have $g \not \in T$ \wp $\nicefrac{2}{3}$. Now condition on neither $f$ nor $g$ in $M$ (happens \wp $\nicefrac{1}{2}$). By \Cref{thm:SRfacts}, 
\[ \nicefrac{1}{3} \leq \ex[f \in T \mid g \not \in T] \]
By \Cref{thm:Hoeff}, $\pr(f \in T \mid g \not \in T) \geq \nicefrac{1}{3}$, so $\pr(f \in T \land g \not \in T) \geq \nf{1}{3} \cdot \nf{2}{3} = \nf{2}{9}$.
Putting all of this together, we get 
\[ \pr(f \in T \land g \not \in T) \geq \nicefrac{1}{4} \cdot \nicefrac{2}{3} + \nicefrac{1}{2} \cdot \nicefrac{2}{9} = \nicefrac{5}{18} \geq \nf{12}{72}. \qedhere \]
\end{proof}

\begin{proof}[Proof of \Cref{clm:odd2}, Even Case] \label{3.2evencase}
\emph{The \MI claims:}  Each perfect matching $M$ contains one of these four edges in
  $\partial v$. Say that edge is $e$. If $e$ also belongs to $M'$ (\wp
  $\nf17$), then $T$ contains exactly one of $\{f,g,h\}$. Since $\nf17
  \geq \nf2{21}$, this gives us the first bound in \Cref{clm:odd2}. Moreover,
  the probability of this edge in $T$
  belonging to the other pair (in this case, $\{g,h\}$) is
  $\nf23$. Hence $\mathbb{P}( |T \cap \{e,f\}| = 1 \land |T \cap \{g,h\}| =
  1) \geq \nf2{21} > \nf4{63}$, giving us the other bound in \Cref{clm:odd2} for the \MI sampler.

  \emph{The \ME claims:}
  For the first bound in \Cref{clm:odd2}, we have that one of the four
  edges must be in $M$ and have $y$-value 1. W.l.o.g., call that edge
  $e$. The other three, $f, g, h$, will be $\nicefrac{1}{3}$-valued
  edges. Since
  $\ex[X_{f, g, h}] = 1$, we may apply \Cref{thm:Hoeff} and obtain a lower bound of
  \[ \mathbb{P}(X_{f, g, h} = 1) \geq 3 \cdot \nicefrac{1}{3} \cdot
    (\nicefrac{2}{3})^2 = \nicefrac{4}{9} \geq \nf{8}{27}. \] For the second bound in
   \Cref{clm:odd2}, call $S_1 = \{e, f\}$ and $S_2 = \{g,
  h\}$. W.l.o.g., we once again label $e \in
  M$. Then, $\mathbb{P}(f \not \in T) = \nicefrac{2}{3}$. So for $S_2$, we
  condition on $f \not \in T$. Then by \Cref{thm:SRfacts},
  \[ \nicefrac{2}{3} \leq \ex[X_{S_2} \mid f \not \in T] \leq
    \nicefrac{2}{3} + \nicefrac{1}{3} = 1. \] Hence, we use
  \Cref{thm:Hoeff} again 
  and obtain $\mathbb{P}(X_{S_2} = 1 \mid f \not \in T) \geq 2 \cdot \nicefrac{1}{3} \cdot \nicefrac{2}{3} = \nicefrac{4}{9}$. 
  In total, we obtain 
  \[ \mathbb{P}(|T \cap \{e, f\}| = |T \cap \{g, h\}| = 1) \geq
    \nicefrac{2}{3} \cdot \nicefrac{4}{9} =
    \nicefrac{8}{27} \geq \nf{16}{81}. \qedhere \]
\end{proof}

\begin{proof}[Proof of \Cref{clm:odd4}a, Even Case] \label{3.3Aevencase}
  \emph{The \MI claims:}
      The event happens when $e \in M'$, which
  happens \wp $\nf1{28}$, which is at least $\nf1{36}$. 
  
 \emph{The \ME claims:} %
  Condition on $e \in M$. Let $S_1 = \partial(u) \setminus e$ and $S_2 = \partial(v) \setminus e$.  Denote $S_1 = \{a, b, c\}$. Lower bound $\mathbb{P}(|S_1 \cap T|=1)$ using  \Cref{thm:Hoeff}: $\mathbb{E}[|S_1 \cap T|]= 3 \cdot \nf13 = 1$, so $ \mathbb{P}(|S_1 \cap T| = 1) \geq 3 \cdot \nicefrac{1}{3} \cdot \left(\nicefrac{2}{3}\right)^{2/3} = \nicefrac{4}{9}.$
Consider the distribution over the edges in $S_2$ conditioned on $a \in T$; this distribution is also SR. By \Cref{thm:SRfacts},
$ \nicefrac{1}{3} \leq \mathbb{E}[X_{S_2} \mid X_a = 1] \leq 1.$ Applying \Cref{thm:SRfacts} twice more, 
\[\nicefrac{1}{3} \leq \mathbb{E}[X_{S_2} \mid X_a = 1, X_{b,c}=0] \leq 1 + \nicefrac{1}{3} + \nicefrac{1}{3} = \nicefrac{5}{3}.\] By \Cref{thm:Hoeff}, $\mathbb{P}(X_{S_2} = 1 \mid X_a = 1, X_{b,c}=0) \geq 3 \cdot \nicefrac{1}{9} \cdot \left(\nicefrac{8}{9}\right)^2 = \nicefrac{64}{243}$.
Using symmetry, we obtain $\mathbb{P}(X_{S_2} = 1 \land X_{S_1} = 1 \land e \in M) \geq \nicefrac{64}{243} \cdot \nicefrac{4}{9} \cdot \nf14 \geq \nf{128}{6561}.$ 
\end{proof}

\begin{proof}[Proof of \Cref{clm:odd4}b, Even Case] \label{3.3bevencase}
  \emph{The \MI claims:}
For part~(b), suppose $e \in M$ (\wp $\nf14$), then each of
  $u,v$ have two other internal edges, each with $y$-value
  $\nf13$. Let us say the good cases are when exactly one of these
  four is chosen; exactly one of $u,v$ has degree $2$ and the other
  has degree $1$ in these cases. We cannot choose zero of these four
  edges, because of the connectivity of $T$, so all bad cases choose
  at least two of these four. Given the $y$-values of $\nf13$ on all
  four edges, the expected number of these edges chosen are $\nf43$,
  so the the probability of a bad case at most $\nf13$. This means
  that with probability at least $\nf14 \cdot (1- \nf13) = \nf1{6}$,
  exactly one of $u,v$ has odd degree in $T$. Since $\nf16 \geq
  \nf19$, we have proven~(b).

  \emph{The \ME claims:}
  Observe that the edge $e$ will always contribute either: 0 to both
  the degree of $u$ in $T$ and the degree of $v$ in $T$ OR 1 to both of these degrees. Let $a,b$ be
  the two internal edges incident to $u$ and $c,d$ be the two edges
  incident to $v$.

  \begin{enumerate}%
  \item \emph{Case 1: $e \in M$.} This implies $a, b, c, d$ are all
    $\nicefrac{1}{3}$-valued edges, so
    $\ex[X_{a,b,c,d}] = \nicefrac{4}{3}$.
    Note that we must have $X_{a,b,c,d} \geq 1$. Therefore,  
    $\mathbb{P}(X_{a, b, c, d} = 1) \geq \nicefrac{2}{3}$.
    
  \item \emph{Case 2: $e  \not \in M$, exactly two of $a,b,c,d$ are in $M$.} W.l.o.g., say $a,c \in M$, so $a$ and $c$ are in $T$. 
    $\ex[X_{b, d}] = \nicefrac{2}{3}$,
    which implies (by \Cref{thm:Hoeff}) that
    $\mathbb{P}(X_{b, d} = 1) \geq \nicefrac{4}{9}$.
    
  \item \emph{Case 3: $e \not \in M$, exactly one of $a,b,c,d$ is in $M$.}
    W.l.o.g., say $a \in M$, so we have 
    $\ex[X_{b, c, d}] = 1$.
    We condition on $e \not \in T$, which happens w.p. \nf23. Since we sample a spanning tree on the internal vertices, this means that at least one of $c,d$ is in $T$. W.l.o.g., say $c \in T$. So in order to bound the probability that the internal parities of $u,v$ are different, equivalently we would like to bound $\pr(X_{b,d}=1 \mid e \not \in T)$. Since 
    $\nicefrac{2}{3} \leq \mathbb{E}[X_{b,d} \mid e \not \in T] \leq \nicefrac{2}{3} + \nicefrac{1}{3}$, 
     \Cref{thm:Hoeff} implies that $\mathbb{P}(X_{b,d}=1 \mid e \not \in T) \geq \nicefrac{4}{9}$. 
    So removing the conditioning on $e \not \in T$ gives a lower bound of $\nf{2}{3} \cdot \nf{4}{9} = \nf{8}{27}$.
  \end{enumerate}
  Taking the minimum of the bounds $\nf23, \nf49$, and $\nf8{27}$ in the three cases gives 
  \[\pr(\text{exactly one of }u,v \text{ has odd degree in }T) \geq \nf{8}{27} \geq \nf{5}{18}. \qedhere\]
\end{proof}

\section{Sampling Algorithm for General Solutions; and Cut Hierarchy}
\label{sec:min-cut-cactus}

Now that we can sample a spanning tree from a graph
with no proper min-cuts, we introduce the algorithm to sample an
$r_0$-tree from a 4-regular, 4EC graph, perhaps with proper min-cuts.

\subsection{The Algorithm}

Assume that the graph $G = (V,E)$ has a set of three special
vertices $\{r_0, u_0, v_0\}$, with each pair $r_0, u_0$ and $r_0, v_0$
having a pair of edges between them. This is without loss of
generality (used in line \ref{endcycle}). Define a \emph{double cycle}
to be a cycle graph in which each edge is replaced by a pair of
parallel edges, and call each such pair \emph{partner edges}.

\begin{algorithm} 
\caption{Sampling Algorithm for a Half-Integral Solution}\label{algo}
\begin{algorithmic}[1]

\State \textbf{let} $G$ be the support graph of a half-integral solution $x$. 
\State \textbf{let} $T = \emptyset$. 
\While{there exists a proper tight set of $G$ that is not crossed by another proper tight set} 
\State \textbf{let} $S$ be a minimal such set (and choose $S$ such that $r_0 \not \in S$). \label{criticalsets}
\State Define $G' = G / (V \setminus S)$. \label{localmulti}
\If{$G'$ is a double cycle}
\State Label $S$ a \emph{cycle set}. 
\State \textbf{sample} a random edge from each set of partner edges in $G[S]$; add these edges to $T$. 
\Else \hspace{0.1cm}// \textit{\small{$G'$ has no proper min-cuts (\Cref{clm:dichotomy})}}. \label{G'sec3form}
\State Label $S$ a \emph{degree set}. If $G' = K_5$, label $S$ a \emph{$K_5$ degree set}, and else a \emph{non-$K_5$ degree set}.
\If{$G'=K_5$} 
\State \textbf{sample} a random path on $G[S]$
\Else 
\State W.p.\ $\lambda$, let $\mu$ be the \ME distribution over $E(S)$ 
\State W.p.\ $1 - \lambda$, let $\mu$ be the \MI distribution. 
\State \textbf{sample} a spanning tree on $G[S]$ from $\mu$ and add its edges to $T$. 
\EndIf
\EndIf
\State \textbf{let} $G = G / S$ \label{algocontract}
\EndWhile
\State Due to $r_0, u_0, v_0$, at this point $G$ is a double cycle
(\Cref{clm:dcycle}). Sample one edge between each pair of adjacent vertices in $G$.  \label{endcycle}
\end{algorithmic}
\end{algorithm}

The sampling algorithm appears as Algorithm \ref{algo}, and is very
similar to that in \cite{KKO19}:
they refer to the sets in line \ref{criticalsets} as \textit{critical sets}. Since $G$ is a 4-regular, 4EC graph at every stage of the algorithm, if $|S| = 2$ or 3, then $S$ must be a cycle set, whereas if $|S| \geq 4$, then $S$ may be a degree or cycle set. 
The following two claims are proved in Appendix \ref{welldefinedalgo}, and show that Algorithm~\ref{algo} is well-defined.

\begin{claim} \label{clm:dcycle}
The graph remaining at the end of the algorithm (line \ref{endcycle}) is a double cycle.
\end{claim}

\begin{claim} \label{clm:dichotomy}
In every iteration in Algorithm~\ref{algo}, $G'$ is either a double cycle or a graph with no proper min-cuts. 
\end{claim}

We will prove the following theorem in \S\ref{sec:charging-argument}. This in turn gives  \Cref{thm:main}. 

\begin{theorem} \label{thm:Ojoinlowering}
  Let $T$ be the $r_0$-tree chosen from Algorithm \ref{algo}, and $O$ be the set of odd degree vertices in $T$. The expected cost of the minimum cost $O$-join for $T$ is at most $(\nf{1}{2} - \varepsilon)\cdot c(x)$. 
\end{theorem}

\subsection{The Cut Hierarchy}

Recall our ultimate goal is to create a low cost, feasible solution in the $O$-join polytope, where $O$ is the set of odd degree vertices in our sampled tree $T$. We start with the fractional solution $z = x/2$ and then reduce the $z_e$ value of some edges. In the process, we may violate some constraints corresponding to min-cuts. To fix these cuts, we need a complete description of the min-cuts of a graph. This is achieved by the implicit hierarchy of critical sets that Algorithm \ref{algo} induces.

The hierarchy is given by a rooted tree
$\cT = (V_{\cT}, E_{\cT})$.\footnote{Since there are several graphs
  under consideration, the vertex set of $G$ is called
  $V_G$. Moreover, for clarity, we refer to elements of $V_G$ as
  vertices, and elements of $V_{\cT}$ as
  nodes.}  The vertex set $V_{\cT}$ corresponds to all critical sets found by the algorithm, along with a root node and leaf nodes representing the vertices in $V_G \setminus \{r_0\}$. If $S$ is a critical set, we label the node in $V_{\cT}$ with $S$, where we view $S \subseteq V_G$ and not $V_{G'}$. The root node is labelled $V_G \setminus \{r_0\}$ and each leaf nodes is labelled by the vertex in $V_G \setminus \{r_0\}$ corresponding to it. Now we define the edge set $E_{\cT}$. A node $S$ is a child of $S'$ if $S \subset S'$ and $S'$ is the first superset of $S$ contracted after $S$ in the algorithm. In addition, the root node is a parent of all nodes corresponding to critical sets that are not strictly contained in any other critical set (i.e., the critical sets corresponding to the vertices in the graph $G$ from line \ref{algocontract} when the \textbf{while} loop terminates). Finally, each leaf node is a child of the smallest critical set that contains it (or if no critical set contains it, is a child of the root node). Thus by construction, vertex sets labelling the children of a node are a partition of the vertex set labelling that node. A node in $V_{\cT}$ is a \emph{cycle} or \emph{degree} node if the corresponding critical set labelling it is a cycle or degree set, respectively. We say the root node is a cycle node (since the graph $G$ in line \ref{endcycle} is a double cycle), and accordingly call $V_G \setminus \{r_0\}$ a cycle set. (The leaf nodes are not labelled as degree or cycle nodes.)

\begin{figure}\centering
\includegraphics[scale=.6]{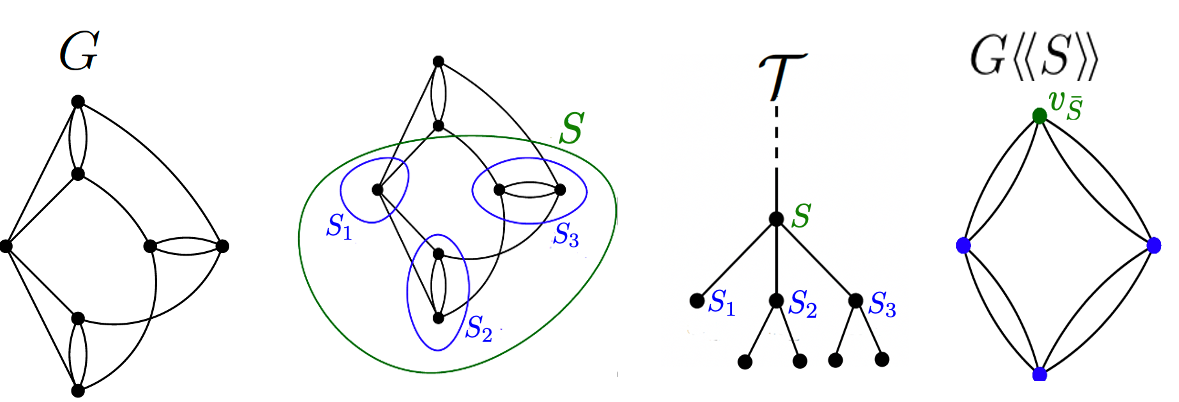}
\caption{A portion of the cut hierarchy $\cT$ and the local multigraph $G\dip{S}$.}
\label{fig:GS}
\end{figure}

\begin{definition}[Local multigraph]
  Let $S \subseteq V_G$ be a set labelling a node in $\cT$. Define the \emph{local multigraph} $G\dip{S}$ to be the following graph: take $G$ and contract the subsets of $V_G$ labelling the children of $S$ in $\cT$ down to single vertices and contract $\bar{S}$ to a single vertex $v_{\bar{S}}$. Remove any self-loops. The vertex $v_{\bar{S}}$ is called the \emph{external vertex}; all other vertices are called \emph{internal vertices.} An internal vertex is called a \emph{boundary vertex} if it is adjacent to the external vertex. The edges in $G\dip{S} \setminus v_{\bar{S}}$ are called \emph{internal edges}. Observe that $G\dip{S}$ is precisely the graph $G'$ in line \ref{localmulti} of Algorithm \ref{algo} when $S$ is a critical set, and is a double cycle when $S = V_G \setminus \{r_0\}$. 
\end{definition}

\textbf{Properties of $\cT$:}
\begin{enumerate}
    \item Let $G$ be a 4-regular, 4EC graph with associated hierarchy $\cT$. Let $S \subseteq V_G$ be a set labelling a node in $\cT$. If $S$ is a degree node in $\cT$, then $G\dip{S}$ has at least five vertices and no proper min-cuts, and hence every proper cut in $G\dip{S}$ has at least 6 edges. If $S$ is a cycle node in $\cT$, then $G\dip{S}$ is exactly a double cycle. 
    
    These follow from \Cref{clm:dichotomy} and the equivalence between $G'$ and $G\dip{S}$. 
    
    \item Algorithm \ref{algo} can be restated as follows: For each non-leaf and non-root node $S$ in $\cT$, sample a random path on $G\dip{S} \setminus v_{\bar{S}}$  if it is a double cycle or $K_5$, and otherwise use the $\ME$ or $\MI$ samplers w.p.\ $\lambda$ and $1-\lambda$, respectively, on $G\dip{S} \setminus v_{\bar{S}}$. Sample a uniformly random cycle on the double cycle in line \ref{endcycle}.
    
    \item For a degree set $S$, the graph $G\dip{S}$ having no proper min-cuts implies that it has no parallel edges. In particular, no vertex has parallel edges to the external vertex in $G\dip{S}$. Hence we get the following:
\end{enumerate}

\begin{corollary}
  \label{clm:boundary-degrees}
  For a set $S$ labeling a non-leaf node in $\cT$ and any internal vertex $v \in G\dip{S}$: if $S$ is a cycle set then $|\partial v \cap \partial S| \in \{0,2\}$, and if $S$ is a degree
  set then $|\partial v \cap \partial S| \in \{0,1\}$.
\end{corollary}

Finally, we show how the hierarchy $\cT$ allows us to characterize the min-cut structure of $G$. The cactus representation of min-cuts (\cite{FF}) is a compact representation of the min-cuts of a graph, and it can be constructed from the cut hierarchy; we defer the details to Appendix \ref{cactusequiv}. In turn we obtain the following complete characterization of the min-cuts of $G$ in terms of local multigraphs.  

\begin{restatable}{claim}{CutMapping}
  \label{lem:min-cuts-mapping}
  Any min-cut in $G$ is either (a)~$\partial S$ for some node $S$ in $\cT$, or 
  (b)~$\partial X$ where $X$ is obtained as
  follows: for some cycle set $S$ in $\cT$, $X$ is the union of 
  vertices corresponding to some contiguous segment of the cycle $G\dip{S}$.
\end{restatable}

\section{Analysis Part I: The Even-at-Last Property}
\label{sec:good-edges}

The proof of \Cref{thm:Ojoinlowering} proceeds in two parts:
\begin{OneLiners}
\item[(1)] In this section, we show that each edge $e$ is ``even-at-last'' with
  constant probability. (This is an extension of the property that both of
  its endpoints have even degree.)

\item[(2)] Then we construct the fractional $O$-join. As always,
  $z = x/2$ is a feasible join, but we show how to save a constant
  fraction of the LP value for an edge when it is even-at-last. This
  savings causes other cuts to be deficient, so other edges \emph{raise}
  their $z$ values in response. However, a charging argument shows
  that the $z$-value for an edge does decrease by a constant factor, in
  expectation. This argument appears in \S\ref{sec:charging-argument}.
\end{OneLiners}

To address part (1), let us define a notion of evenness for every edge in $G$. In the case where $G$ has no proper min-cuts, we called an edge \emph{even} if both of its endpoints were even in $T$. Now, the general definition of evenness will depend on where an edge belongs in the hierarchy $\cT$.

\begin{definition}
  We say an edge $e \in E(G)$ is \emph{settled} at $S$ if $S$ is the (unique) set such that $e$ is an internal edge of $G\dip{S}$; call $S$ the \emph{last set} of $e$. If $S$ is a degree or cycle set, we call $e$ a \emph{degree edge} or \emph{cycle edge}, respectively. 
\end{definition}

\begin{definition}[Even-at-Last]
Let $S$ be the last set of $e$, and $T\dip{S}$ be the restriction of $T$ to $G\dip{S}$. 

\begin{enumerate}
\item A degree edge  $e$ is called \emph{even-at-last (EAL)} if both of
  its endpoints have even degree in $T\dip{S}$.
\item For a cycle edge $e = uv$, the graph $G\dip{S} \setminus \{v_{\bar
    S}\}$
  is a chain of
  vertices $v_\ell, \ldots, u, v, \ldots, v_r$, with consecutive
  vertices connected by two parallel edges. Let
  $C := \{v_\ell, \ldots, u\}$, and $C' := \{v, \ldots, v_r\}$ be the
  partition of this chain. The cuts $\partial C$ and $\partial C'$
  are called the \emph{canonical cuts} for $e$. Cycle edge $e$ is called \emph{even-at-last
    (EAL)} if both canonical cuts are crossed an even number of
  times by $T\dip{S}$; in other words, if there is exactly one edge in $T\dip{S}$ from each of the two pairs of external partner edges leaving $v_{\ell}$ and $v_r$. 
\end{enumerate}
\end{definition}

Informally, a degree edge is EAL in the general case if it is even in the tree at the level at which it is settled. Also note that \emph{cycle} edges settled at the same (cycle) set are either all EAL or none are EAL. 

\begin{definition}[Special and Half-Special Edges]
  Let $e$ be settled at a \emph{degree} set $S$. We say that $e$ is
  \emph{special} if both of its endpoints are non-boundary vertices in
  $\contract$ and \emph{half-special} if exactly one of its endpoints
  is a boundary vertex in $\contract$.
\end{definition}

We now prove a key property used in \S\ref{sec:charging-argument} to reduce the $z$-values of edges in the fractional $O$-join.

\begin{theorem}[The Even-at-Last Property]
  \label{thm:EAL}
   The table below gives lower bounds on the probability that
   special, half-special, and all other types of degree edges are EAL in each
   of the two samplers.  
   \begin{center}
    \begin{tabular}{ c c c c}
     & special & half-special & other degree edges \\ 
     \MI & \nf{1}{36} & \nf{1}{21} & \nf{1}{18} \\  
     \ME & \nf{128}{6561} & \nf{4}{27} &  \nf{1}{9} \\
    $K_5$ & - & - & \nf{1}{9}
    \end{tabular}
  \end{center}
  Moreover, a cycle edge is EAL w.p.\ at least $\lambda \cdot \nf19 + (1-\lambda) \cdot \nf{4}{63}$.
\end{theorem}

\begin{proof} 
  Let $e$ be settled at $S$. Let $T_S$ be the spanning tree sampled on the internal vertices of $G\dip{S}$ (in the notation of Algorithm \ref{algo}, the spanning tree sampled on $G[S]$).

  First, assume that $S$ is a degree set: 
  \begin{enumerate}
  \item If none of the endpoints of $e$ are boundary vertices in $G\dip{S}$ (i.e., $e$ is special),
    then it is EAL exactly when both its endpoints have even degree in
    $T_S$. By
    \Cref{clm:odd4}(a), this happens \wp
    $\nf{1}{36}$ for the \MI sampler and \wp $\nf{128}{6561}$ for the \ME sampler.

  \item Now suppose that $e$ is half-special, so that exactly one of the endpoints of $e = uv$ (say $u$) is a boundary
    vertex in $G\dip{S}$, with edge $f$ incident to $u$ leaving $S$.
    By \Cref{clm:odd2}, the other endpoint $v$ is even in $T_S$ \wp~$\nf2{21}$ for the \MI sampler and $\nf{8}{27}$ for the \ME sampler. Moreover,
    the edge $f$ is chosen at a higher level than $S$ and is therefore
    independent of $T_S$, and hence can make the degree of $u$ even
    \wp $\nf12$. Thus $e$ is EAL \wp $\nf{1}{21}$ for the \MI sampler and $\nf{4}{27}$ for the \ME sampler. 

  \item Suppose both endpoints of $e$ are boundary
    vertices of $S$, with edges $f,g$ leaving $S$. Let $q_=$ be the probability
    that the degrees of vertices $u,v$ in the tree $T_S$ chosen within
    $S$ have the same parity, and $q_{\neq} = 1 - q_=$. Now, when $S$
    is contracted and we choose a $r_0$-tree $T'$ on the graph $G/S$
    consistent with the marginals, let $p_=$ be the probability that
    either both or neither of $f,g$ are chosen in $T'$, and
    $p_{\neq} = 1 - p_=$. Hence
    \begin{gather}
      \Pr[e \; EAL] = q_{oo}p_{11} + q_{oe}p_{10} + q_{eo}p_{01} +
      q_{ee}p_{00} = \nf12 (p_= q_= + p_{\neq} q_{\neq}), \label{eq:quad}
    \end{gather}
          where $q_{oo}, q_{oe}, q_{eo}$ and $q_{ee}$ correspond to different parity
          combinations of $u$ and $v$ in $T_S$ and $p_{00}, p_{01}, p_{10}, p_{11}$ 
          correspond to whether $f$ and $g$ are chosen in $T'$. The second inequality 
          follows from \Cref{claim:symmetry} applied to the random $r_0$-tree $T'$.
    \begin{enumerate}
    \item If $f,g$ are settled at different levels, then they are
      independent. This gives $p_= = p_{\neq} = \nf12$, and hence
      $\Pr[ e \; EAL ] = \nf14$ regardless of the sampler.
    \item If $f,g$ have the same last set which is a degree set, first consider when $S$ is a non-$K_5$ degree set. There are two cases. If the last set is not a $K_5$, then by \Cref{clm:odd1} each of the quantities
      $p_{11}, p_{01}, p_{10} \geq \nf19$. %
        By \Cref{claim:symmetry}, $p_{00}\ge\nf19$ as well. If the last set is a $K_5$, then $p_{11} = 1/6$, $p_{01} = p_{10} = 1/3$,  and $p_{00} = 1/6$. Hence, (\ref{eq:quad}) gives
      $\Pr[ e \; EAL] \geq \nf19$.
    \item If $f,g$ have the same last set which is a cycle set, then
      first consider the case where $f,g$ are partners, 
        in which case
      $p_{\neq} = 1$. Now~(\ref{eq:quad}) implies that
      $\Pr[e \; EAL] = \nf{q_{\neq}}{2}$, which by~\Cref{clm:odd4}(b) is at least
      $ \nf12 \cdot \nf1{9} = \nf{1}{18}$ in the \MI sampler, at least $\nf12 \cdot \nf5{18} = \nf5{36}$ in the \ME sampler, and exactly $\nf{1}{2} \cdot \nf{1}{3} = \nf{1}{6}$ in the $K_5$ sampler.
      
      If $f,g$ are not partners, then they are chosen independently, in
      which case again $p_= = p_{\neq} = \nf12$, and hence
      $\Pr[ e \; EAL ] = \nf14$.
    \end{enumerate}
  \end{enumerate}    
  
  Next, let $S$ be a cycle set. Let $e = uv$ be an edge inside the
  cycle, and let $\{a,b\},\{c,d\}$ be the four edges crossing
  $\partial S$. Let $v_S$ be the vertex obtained by contracting down
  $S$ (whose incident edges are then $\{a,b,c,d\}$). Now in order for
  $e$ to be EAL, one each of $\{a,b\}$ and $\{c,d\}$ must belong to
  $T$; call this event $\cE$. We again consider cases based on where
  these edges are settled.  Let node $P$ be the parent of node $S$,
  and let $v_S$ be the vertex in $G\dip{P}$ obtained from contracting $S$.
  
  \begin{enumerate}
  \item If all four edges are settled at $P$, and $P$ is a degree set,
    then in particular $P$ is a non-$K_5$ degree set. Then, %
     \Cref{clm:odd2} says that for the \MI sampler, we have $\Pr[\cE] \geq \nf4{63}$.
    In contrast, the \ME sampler gives us $\Pr[\cE] \geq \nicefrac{16}{81}$. 
    
  \item If all four edges are settled at $P$, and $P$ is a cycle set,
    then no matter how these four edges are distributed, $\Pr[\cE]
    \geq \frac12$.
  \item If three of them $\{a,b,c\}$ are settled at $P$, and the
    fourth (say $d$) at a higher level, then since $P$ has a vertex
    $v_S$ with a single edge leaving it, $P$ must be a degree set.
    Now exactly one of $\{a,b\}$ is chosen in $T$ \wp (exactly) $\nf{1}{3}$ in the $K_5$ sampler, and \wp at least $\nf29$ in the \MI sampler and $\nf13$ in the \ME sampler,
    by \Cref{clm:odd1}. And since $d$ is
    independently picked at a different level, exactly one of
    $\{c,d\}$ is chosen in $T$ \wp $\nf12$, giving an overall
    probability of $\nf19$ in the \MI sampler and $\nf16$ in the \ME sampler.

  \item Finally, if only two edges are settled at $P$, and two others go to
    higher levels, then $P$ is a cycle set by \Cref{clm:boundary-degrees}. In this case, exactly one
    of the two edges that are settled in $P$ is chosen. Now we want a
    specific one of the edges going to a higher level to be chosen
    into $T$ (and the other to not be chosen), which in the worst case
    happens \wp at least $\nf19$, due
    to \Cref{clm:odd1} (and using that the analogous bound for the $K_5$ sampler is $\nf{1}{3}$, which is only better than $\nf{1}{9}$).
  \end{enumerate}
\end{proof}

\section{Analysis Part II: The $O$-Join and Charging}
\label{sec:charging-argument}

To prove \Cref{thm:Ojoinlowering} and thereby finish the proof, we
construct an $O$-join for the random tree $T$, and bound its expected
cost via a charging argument. The structure here is similar to \cite{KKO19}; however, we use a flow-based
argument to perform the charging instead of the naive one, and also use our stronger notion of evenness (EAL). 

Let $O$ denote the (random set of) odd-degree vertices in $T$. The dominant of the $O$-join polytope
$\Kjoin(G,O)$ is given by
\[ x(\partial S) \geq 1 \qquad \qquad \forall S \sse V, |S \cap O|
  \odd. \] This polytope is integral, so it suffices to exhibit a
fractional $O$-join solution $z \in \Kjoin(G,O)$, with low expected cost. 
(The expectation is taken over $O$.) Note that $|S \cap O|
  \odd$ if and only if $|\partial S \cap T| \odd$.
  
Now \Cref{thm:Ojoinlowering} follows from the claim below, which we will prove in this section. 

\begin{lemma} \label{lem:O-join decrease}
  There is an $\eps > 0$ such that if the $r_0$-tree $T$ is sampled using the procedure described in
  Algorithm \ref{algo}, and $O$ is the set of odd-degree vertices in $T$, then there is
  a fractional solution $z \in \Kjoin(G,O)$ where $\E[z_e] \leq (\nf12 -\eps) x_e$
  for all edges $e \in E(G)$.
\end{lemma}

\subsection{Construction of the Fractional $O$-join}

The construction of the fractional $O$-join $z$ goes as follows: We
start with the solution $z=\nf{x}2$. Notice that
$z(\partial S) \geq 1$ is a tight constraint in this initial solution
when $S$ is a min-cut.  Now we describe how to reduce the $z_e$
values.

Define 
\[p_{sp}^{MI} \coloneqq \nf{1}{36}, p_{hs}^{MI} \coloneqq \nf{1}{21}, p^{MI} \coloneqq \nf{1}{18}, p_{sp}^{ME} \coloneqq \nf{128}{6561}, p^{ME} \coloneqq  \nf19\]
\[p_{sp} \coloneqq \lambda p_{sp}^{ME} + (1-\lambda)p_{sp}^{MI}, p_{hs} \coloneqq \lambda p^{ME} + (1-\lambda)p_{hs}^{MI}, p \coloneqq \lambda p^{ME} + (1-\lambda)p^{MI}. \]
 Let $p_{sp}(e), p_{hs}(e),$ and $p_d(e)$ denote the probabilities that $e$ is EAL if $e$ is a special degree edge, half-special degree edge, or other degree edge, respectively. Let $p_c(e)$ denote the probability that a cycle edge $e$ is EAL. By \Cref{thm:EAL}, 
\[p_{sp}(e) \geq p_{sp},\ p_{hs}(e) \geq p_{hs},\ p_d(e) \geq p,\ p_c(e) \geq p.  \]

(We do not distinguish the
half-special case in the \ME sampler, as the half-special bound of $\nf{4}{27}$ in \Cref{thm:EAL} is greater than $p^{ME}$.)
Call an edge $e$ a \emph{$K_5$ degree edge} if $e$ is settled at a
degree set $S$ where $G\dip{S}$ is a $K_5$. By the inequalities above, the random variables below are well-defined.  

\begin{enumerate}
    \item Define a Bernoulli random variable for each edge $e$:
    \begin{enumerate}
        \item If $e$ is a special degree edge, set $B_e \sim \Ber(\nf{p_{sp}}{p_{sp}(e)})$. 
        \item If $e$ is a half-special degree edge, set $B_e \sim \Ber(\nf{p_{hs}}{p_{hs}(e)})$. 
        \item If $e$ is any other degree edge, set $B_e \sim \Ber(\nf{p}{p_d(e)})$.
        \item If $e$ is a cycle edge, set $B_e \sim \Ber(\nf{p}{p_c(e)})$. Further, if $e$ and $f$ are partners, perfectly correlate their coin flips, i.e., set $B_e = B_f$. 
    \end{enumerate}
    \item If $e$ is EAL \emph{and} $B_e = 1$, reduce $z_e$ by 
    \begin{enumerate}
        \item $\tau$ if $e$ is a non-$K_5$ degree edge.
        \item $\gamma$ if $e$ is a $K_5$ degree edge. 
        \item $\beta$ if $e$ is a cycle edge. 
    \end{enumerate}
    \item We enforce that $\tau \leq \gamma \leq \beta \leq \nf{1}{12}$, $\beta \geq 2\tau$, and $\beta \geq 2\gamma$. We will optimize $\tau, \gamma,$ and $\beta$ via a linear program later.
\end{enumerate}

The purpose of the Bernoulli coin flips is to flatten the probability that an edge is reduced down to the lower bound on the probability that it is EAL from \Cref{thm:EAL}:
\begin{observation} \label{obs:exactreductions}
    If $e$ is a special, half-special, other degree, or cycle edge, then $z_e$ is reduced with probability exactly $p_{sp}$, $p_{hs}$, $p$, or $p$, respectively. 
\end{observation}

 This reduction scheme may make $z$ infeasible for the $O$-join polytope. We now discuss how to maintain feasiblity.

\subsection{Maintaining Feasibility of the Fractional $O$-join via Charging}

 Suppose $f$ is EAL and that we reduce edge $z_f$ (per its coin flip $B_f$). Say $f$ is settled at $S$. If $S$ is a degree set, then the only min-cuts of $\contract$ are the degree cuts. So the only min-cuts that the edge $f$ is part of in $\contract$ are the degree cuts of its endpoints, call them $U,V$, in $\contract$ ($U$ and $V$ are vertices in $\contract$ representing sets $U$ and $V$ in $G$). But since $|\partial U \cap T|$ and $|\partial V \cap T|$ are both even by definition of EAL,  we need not worry that reducing $z_f$ causes $z(\partial U) \geq 1$ and $z(\partial V) \geq 1$ to be violated. Likewise, if $S$ is a cycle set, then by definition of EAL all min-cuts $S'$ in $\contract$ containing $e$ have $|\partial S' \cap T|$ even, so again we need not worry.

Since $f$ is only an internal edge for its last set $S$, the only cuts $S'$ for which the constraint $z(\partial S') \geq 1$, $|\partial S' \cap T|$ odd, may be violated as a result of reducing $z_f$ are cuts represented in lower levels of the hierarchy. Specifically, lef $f$ be an external edge for some $G\dip{X}$ (meaning $X$ is lower in the hierarchy $\cT$ than $S$) and $S'$ be a min-cut of $G\dip{X}$ (either a degree cut or a canonical cut). By \Cref{lem:min-cuts-mapping}, cuts of the form $S'$ are the only cuts that may be deficient as a result of reducing $z_f$. Call the internal edges of $\partial S'$ \emph{lower edges}. When $z_f$ is reduced and $|\partial S' \cap T|$ is odd, we must distribute an increase (charge) over the lower edges totalling the amount by which $z_f$ is reduced, so that $z(\partial S') = 1$. We say these lower edges receive a \emph{charge} from $f$. How the charge is distributed will depend on whether the lower edges are degree edges or cycle edges (see \S\ref{sec:cycle-edges} and \S\ref{sec:degree-edges} below). 

We claim that this procedure maintains feasibility of the $O$-join solution $z$. Indeed, by \Cref{lem:min-cuts-mapping}, no constraint corresponding to a \emph{min-cut} in the $O$-join polytope is violated. Further, by capping the amount an edge can be reduced at $\nf{1}{4} - \nf{1}{6} = \nf{1}{12}$, we will ensure that \emph{none} of the constraints are violated, since every other cut has at least 6 edges.

To prove \Cref{lem:O-join decrease}, we will lower bound the expected decrease to $z_e$, using different strategies for charging $e$ when $e$ is a cycle edge versus a degree edge. In particular, when $e$ is a cycle edge, we distribute charge from an external edge evenly between $e$ and its partner. When $e$ is a degree edge, charge will be distributed according to a maximum-flow solution.

\subsection{Charging of Cycle Edges}
\label{sec:cycle-edges}

Throughout, let $X_e = 1$ if edge $e$ is in $T$ and 0 otherwise.

We now analyze the expected net decrease for each edge, starting with cycle edges. The expected decrease for a cycle edge starts off as exactly $p \cdot \beta$ by \Cref{obs:exactreductions}. However, we need to control the
amount of charge the edge receives from edges settled at higher
levels. Note that in the calculations, we assume the worst-case
scenario that such edges settled at higher levels are reduced w.p.\ exactly $p$. The
special/half-special edges are reduced with lower probability,
which is only better for us because the expected amount of charge will
be lower in these cases. Note also that throughout the calculations we may assume that an edge settled at a higher level is a cycle edge. This will follow from the assumption that $\beta \geq 2\tau, 2\gamma$.

Consider a cycle edge $e$ that is settled at set $S$. Let its
partner on the cycle be $e'$, and let the four external edges for $S$
be $\{a,b\}$ (leaving the vertex $u$) and $\{c,d\}$ (leaving vertex
$v$). When $S$ is contracted, vertex $v_S$ has these four edges
$\{a,b,c,d\}$ incident to it. Let $P$ be $S$'s parent node in
$\cT$. Moreover, let $C$ and $C'$ be the canonical cuts, as in the
figure below. We consider the following cases as in \cite{KKO19}.

\begin{figure}[!htb]
\minipage{0.45\textwidth}
  \includegraphics[width=\linewidth, page=1]{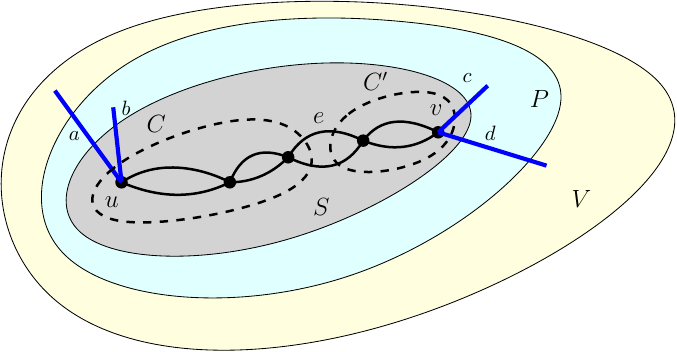}
  \caption{Case (i)}\label{fig:chargecase1}
\endminipage\hfill
\minipage{0.45\textwidth}
  \includegraphics[width=\linewidth, page=2]{charge2}
  \caption{Case (ii)}\label{fig:chargecase2}
\endminipage
\end{figure}

\begin{enumerate}
\item[(i)] Two of the four edges incident to $v_S$ also leave
  $P$---and therefore $P$ is a cycle set by
  \Cref{clm:boundary-degrees}. Of the two edges leaving $P$, one must
  be from $\{a,b\}$ and one from $\{c,d\}$. To see this, say both $c$ and $d$ (leaving node $u$) both did not leave $P$. Consider the set of vertices $S'$ which $c$ and $d$ lead to. $S' \cup u$ is a proper min cut which crosses $S$; however, this contradicts the fact that $S$ is a critical set (namely $S$ is not supposed to be crossed by another proper min cut).

  Suppose $\{a,d\}$ are the external edges
  for $P$, and $\{b,c\}$ are internal (and hence partners in this cycle set).

  Now, consider the event that $b$ is reduced; the important
  observation is that $c$ is also reduced, because both are EAL at the
  same time and their coins are perfectly correlated. Not only that, the event that $b$ is reduced implies that $|\delta(S) \cap T|$ is even. This means that
  both cuts $C, C'$ are deficient by the same amount and at the same time, due to this event
  of reducing $b,c$, and raising $e$ helps both of them. This means
  the net expected charge to $e$ is at most
  \[ p\frac\beta2 \big( \Pr[ C \odd \mid b \reduced] + \Pr[ C \odd
    \mid a \reduced] + \Pr[ C' \odd \mid d \reduced] \big). \] (Note
  that there is no term for $c$ here because of the discussion above;
  moreover, the $\nf\beta2$ term reflects that the charge of $\beta$
        is split between $e,e'$.) 

First we argue that $\Pr[C \odd \mid a \reduced] = \nf12$ (and by the same argument we then have $\Pr[C' \odd \mid d \reduced] = \nf12$). To see this, note that the event that $C$ is odd depends only on which of $a$ and $b$ are in $T$. Moreover, the event that $b$ is in $T$ is independent of the event that $a$ is reduced, since $b$ is contracted in the last set of $a$. Hence $\Pr[C \odd \mid a \reduced] = \Pr[X_a = X_b] = \nf{1}{2}$, as $a$ and $b$ are settled at different sets.

Next we argue that $\Pr[C \odd \mid b \reduced] = \nf12$. If $b$ is reduced, then $b$ is EAL, which implies $|T \cap \{a,d\}| = 1$. On the other hand, since $b,c$ are partner edges, we have that $|T \cap \{b,c\}| = 1$. So, the only ways for $C$ to be odd given $b$ is reduced are if $a,b \in T$, or if $c,d \in T$. Thus we have:

\begin{align*}
    \Pr[C \odd \mid b \reduced] &= \Pr[X_a = X_b = 1 \mid b \reduced] + \Pr[X_c = X_d = 1 \mid c \reduced] \\
    &= \Pr[X_b = 1] \cdot \Pr[X_a = 1 \mid b \reduced] + \Pr[X_c = 1] \cdot \Pr[X_d = 1 \mid c \reduced] \\
    &= \frac{1}{2} \cdot (\Pr[X_a = 1 \mid b \reduced] + \Pr[X_d = 1 \mid b \reduced]) \\
    &= \frac{1}{2}.
\end{align*}

In the first and third lines, we have used that $b$ is reduced if and only if $c$ is reduced, since $b$ and $c$ are partners. In the second line, we have used that the event $\{X_b = 1 \}$ is independent of the event $\{X_a = 1, b \reduced\}$. This is because the latter event depends only on edges settled at higher levels than $b$, as $b$ is settled at a cycle set, and on the independent coin flip $B_b$. Finally, in the fourth line we have used that conditioned on $b$ being reduced, the events $X_a = 1$ and $X_d =1 $ are disjoint.

  Thus the net expected decrease for edge $e$ is at least 
  
  \begin{gather}
     p(\beta - \nf\beta2 \cdot 3 \cdot \nf{1}2) = \boxed{p \nf\beta4}. \label{eq:cyc1}
  \end{gather}
\item[(ii)] Only one of the edges incident to $v_S$, say edge $a$, is external for $P$. By
  \Cref{clm:boundary-degrees}, $P$ is a degree set, and in the worst case, a $K_5$ degree set since $\gamma \geq \tau$. 
  The charge is maximized when $a$ is a cycle edge, in which case the
  expected charge to $e$ is at most
  \[ p \big( \nf\beta2 \Pr[ C \odd \mid a \reduced ] + \nf\gamma2 \Pr[ C
    \odd \mid b \reduced ] \cdot 1 + 2\cdot \nf\gamma2 \big). \]
where the latter charge is from $c$ and $d$. We see that $\Pr[C \odd \mid a \reduced] = \nf12$, by the same reasoning as in (i) above. We upper bound $\Pr[C \odd \mid b \reduced]$ by 1. Thus the net expected decrease is at least. 
  \begin{gather}
    p(\beta - \nf\beta4 - \nf{3\gamma}{2}) = \boxed{p(\nf{3\beta}4 -
     \nf{3\gamma}{2})}. \label{eq:cyc2}
  \end{gather}

\begin{figure}[!htb]
\minipage{0.45\textwidth}
  \includegraphics[width=\linewidth, page=3]{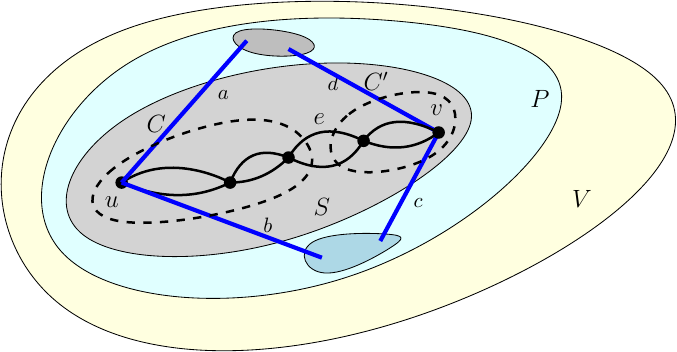}
  \caption{Case (iii)}\label{fig:chargecase3}
\endminipage\hfill
\minipage{0.45\textwidth}%
  \includegraphics[width=\linewidth, page=4]{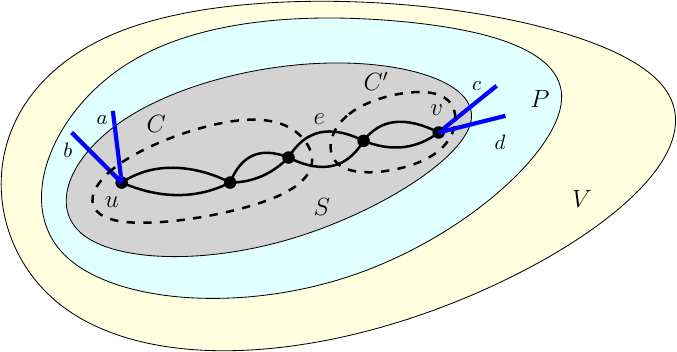}
  \caption{Case (iv)}\label{fig:chargecase4}
\endminipage
\end{figure}

\item[(iii)] All four of the edges are internal to $P$, which is a
  cycle set.
  The partners must now be one edge from each of the two pairs
  $\{a,b\}$ and $\{c,d\}$---let's say $\{a,d\}$ and $\{b,c\}$ are
  partners (using an argument similar to that in (i)). The cut $C$ is odd in $S$ if and only if $C'$ is odd;
  indeed, $C$ is odd means we choose both or neither of
  $\{a,b\}$, and since $\{c,d\}$ are the partners of $\{a,b\}$, we
  choose both or neither of them. Hence, increasing $e$ fixes both
  cuts at the same time, similar to the argument in~(i), so we focus
  only on the increase due to cut $C$. Thus the net expected charge to $e$ is at most
  \[p \cdot \frac{\beta}{2} \cdot (\Pr(C \odd \mid a \reduced) + \Pr(C \odd \mid b \reduced)).\]
  We claim that $\Pr(C \odd \mid a \reduced) = \nf12$ (and likewise for the second term). We have $\Pr(C \odd \mid a \reduced) = \Pr(X_a = X_b \mid a \reduced) = \Pr(X_a = X_b) = \nf12$, where the penultimate equality follows from the fact that $a,b$ are settled at a cycle set, so the event that $a$ is reduced only depends on edges settled at higher levels.  So the net expected decrease is 
  $p(\beta - 2\cdot\nf\beta2\cdot\nf12 ) = p\nf\beta2$, which is
  no worse than~(\ref{eq:cyc1}).
  
\item[(iv)] All four of the edges are internal to $P$, and $P$ is a
  degree set. Note that $G\dip{P}$ cannot be a $K_5$ since the vertex $v_S$ has four internal edges incident to it. 
  We simply bound the charge due to each edge by $\tau/2$, and hence
  get net expected decrease of
  \begin{gather}
    p(\beta - 4\cdot \nf\tau2) = \boxed{p(\beta - 2\tau)}, \label{eq:cyc4}
  \end{gather}
\end{enumerate}

\subsection{Charging Degree Edges: Max-Flow Formulation} \label{sec:degree-edges}

So far we have shown that no cycle edge will receive too much charge. We did this by considering four different configurations for the external edges of $S$, and distributing charge evenly between a cycle edge and its partner. Before showing that no degree edge receives too much charge, we will define a charging scheme for degree edges that achieves better bounds than distributing charge uniformly. 

Let $\contract$ be the local multigraph for some degree set $S$, so $\contract$ is 4-regular and has no proper min-cuts. For an external edge $e$, let $u_e$ denote the internal boundary vertex to which $e$ is incident. We create a bipartite graph $H=(B,F,E)$ from $\contract$ as follows. The vertices of $B$ are labelled with the external edges of $\contract$. The vertices of $F$ are labelled with the internal edges of $\contract$. So $|B|=4$. Place an edge between $e \in B$ and $f \in F$ in the edge set $E$ if $e$ and $f$ share $u_e$ as an endpoint. Set $b_e = 1$ for all $e \in B$. We call $H$ the \emph{bipartization} of $\contract$.

The following lemma follows from the Max-Flow Min-Cut theorem:
\begin{lemma}\label{hall}
  Given a bipartite graph $G=(B,F,E)$ with positive integers $b_u$ for
  each $u \in B$, and $c \geq 0$, there exists
  $x: B \times F \rightarrow \mathbb{R}_{+}$ satisfying
  \begin{align}
    \textstyle \sum_{f \in \partial u} x(u,f)
    &= b_u \text{ for all }  u \in B, \label{constraint1} \\
    \textstyle \sum_{u \in \partial f} x(u,f)
    &\leq c \text{ for all } f \in F \label{upper bound}
  \end{align}
  if and only if 
  \begin{equation}\label{Hall cond}
    \textstyle    |N(R)| \geq \frac{b(R)}{c} \text{ for all } R \subseteq B
  \end{equation}
  where $b(R) = \sum_{u \in R} b_u$.
\end{lemma}

\Cref{hall} now gives the following result, which follows from straightforward casework. 
\begin{lemma}\label{Hall's constants}
  Let $H=(B,F,E)$ be the bipartization of $\contract$, as described
  above. If $\contract$ is not a $K_5$, then the smallest $c$
  satisfying (\ref{Hall cond}) is given by $c=\nf{1}{2}$. If
  $\contract$ is a $K_5$, then the smallest $c$ satisfying (\ref{Hall
    cond}) is given by $c=\nf{2}{3}$.
\end{lemma}

The fact that we may achieve a value of $c = \nf12$ instead of $c = \nf23$ as long as $\contract$ is not a $K_5$ motivates why we utilize a different sampler for $K_5$. In fact, $c = \nf23$ on $K_5$ is achieved by distributing uniform charge over internal edges.

\subsection{Charging of non-$K_5$ Degree Edges}

We now analyze the expected net decrease for degree edges. By \Cref{obs:exactreductions}, the expected decrease for a non-special degree edge starts off as at least
$p_{hs} \cdot \tau$ or $p_{hs} \cdot \gamma$, since $p_{hs} \leq p$.
We will analyze the charge on non-$K_5$ degree and $K_5$ degree edges separately, as a different sampler is used when $G\dip{S} = K_5$. We begin with the former.

\subsubsection{Charging of non-special edges} \label{sec:nonsp}
\begin{enumerate}
    \item \textbf{Case 1: All external edges are non-$K_5$ degree edges.}
Let $x: B \times F \rightarrow \mathbb{R}_+$ be any function satisfying the constraints in \Cref{hall} for $c=\nf{1}{2}$ (by \Cref{Hall's constants} such an $x$ exists, and can be found efficiently using max-flow min-cut). The charging scheme is as follows: for each external edge $e$ and boundary vertex $u_e$, if the cut $\delta(u)$ is odd and the edge $e$ is reduced, charge each internal edge $f$ incident to $u_e$ by $\tau \cdot x(e,f)$. The constraint (\ref{constraint1}) ensures that the charge neutralizes the reduction of $e$ by $\tau$. For any internal edge $f$, let $\delta(f)$ denote the external edges with which $f$ shares an endpoint, i.e., the neighbors of $f$ in the bipartization $H = (B,F,E)$. We have 
\begin{align}
\mathbb{E}[\text{charge to } f] &= \sum_{e \in \delta(f)} \tau \cdot x(e,f) \cdot \mathbb{P}(u_e \odd \wedge e \red) \nonumber \\
& \leq \tau \cdot c \cdot \mathbb{P}(u_e \odd \mid e \red) \cdot \mathbb{P}(e \red) \label{chargeeq} 
\end{align}
Using the naive bound that $\pr(u_e \odd \mid e \red) \leq 1$, we obtain that the the expected decrease is at least 
$\boxed{p_{hs}\tau - \frac{\tau}{2} \cdot p.}$
\item \textbf{Case 2: All external edges are $K_5$ degree edges.}
An analogous argument shows that the expected decrease is at least 
$\boxed{p_{hs}\tau - \frac{\gamma}{2} \cdot p.}$

\item \textbf{Case 3: All external edges are cycle edges.} We have that $\mathbb{P}(u_e \odd \mid e \red) = \frac{1}{2}$ in this case. This is because
\[\Pr(u_e \odd \mid e \reduced) = \Pr\left(X_e \neq \sum_{f \in \delta(e)} X_f \mid e \reduced \right) = \Pr\left(X_e \neq \sum_{f \in \delta(e)} X_f\right) = \frac{1}{2}. \]
where in the penultimate equality we have used that the event that $e$ is reduced is independent of $X_e$, since $e$ is settled at a cycle set (and independent of $X_f$ for $f \in \delta(e)$, since $e$ is settled at a higher level than such $f$). So the expected decrease is at least 
$\boxed{p_{hs}\tau - \frac{\beta}{2} \cdot \frac{1}{2} \cdot p.}$

\item \textbf{Case 4: The external edges are a combination of non-$K_5$ degree edges, $K_5$ degree edges, and cycle edges.} In this case, we use an averaging argument to show that the expected charge is no worse than $p \cdot \max\{\frac{\tau}{2}, \frac{\gamma}{2}, \frac{\beta}{4}\}.$ See Appendix \ref{sec:charging proofs} for details. 
\end{enumerate}

\subsubsection{Charging of special edges}
These edges are never charged, as their endpoints are both non-boundary. Also, these cannot be $K_5$ degree edges by definition. So the expected decrease is equal to the expected reduction, which is at least
    $ \boxed{p_{sp}\tau = \frac{\tau}{36}.}$

\subsection{Charging of $K_5$ Degree Edges}
$K_5$ degree edges have an initial decrease of exactly $p\gamma$, as they are neither special nor half-special. We will use the following lemma.

\begin{lemma} \label{lem:k5oddstuff}
Let $G\dip{S}=K_5$. Let $f$ be an external edge and $u$ be the boundary vertex that is one of its endpoints. Then
\[Pr(u \odd \mid f \red) = \frac{1}{2}. \]
\end{lemma}

\begin{proof}
We note that $\delta(f)$ is independent of the events that $f \in T$ and $f$ is reduced, since $f$ is settled at a higher level than the edges in $\delta(f)$. So,
\begin{align*}
    \Pr[u \odd \mid f \reduced] &= \Pr[X_f  = 1 \wedge \delta(f) \text{ even}   \mid f \reduced] + \Pr[X_f  = 0 \wedge \delta(f) \odd   \mid f \reduced] \\
    &= \Pr[\delta(f) \text{ even}] \cdot \Pr[X_f = 1 \mid f \reduced] + \Pr[\delta(f) \odd] \cdot \Pr[X_f = 0 \mid f \reduced] \\
    &= \frac{1}{2} \cdot (\Pr[X_f = 1 \mid f \reduced] + \Pr[X_f = 0 \mid f \reduced]) \\
    &= \frac{1}{2}
\end{align*}
where in the penultimate line we have used that we sample a path on the internal vertices of the $K_5$.
\end{proof}

We will factor in this probability into the charging of $K_5$ degree edges below. 
\begin{enumerate}
    \item \textbf{All external edges are non-$K_5$ degree edges.}  Since the edge being charged is a $K_5$ edge, \Cref{Hall's constants} gives us a worst-case expected charge of $p\cdot \frac{2\tau}{3}$. However, factoring in $Pr(u \odd \mid f \red)=\frac{1}{2}$ for every external edge $f$, we have a worst-case expected charge of $p \cdot \frac{2\tau}{3} \cdot \frac{1}{2}$, which by (\ref{chargeeq}) gives us an expected net decrease of at least 
    \[\boxed{p\gamma - p\frac{\tau}{3}.} \]
    \item \textbf{All external edges are $K_5$ degree edges.} Using similar reasoning as above, we have a worst-case expected charge of $p \cdot \frac{2\gamma}{3} \cdot \frac{1}{2}$, which gives us an expected net decrease of at least
    \[\boxed{p\gamma - p\frac{\gamma}{3}.} \]
    \item \textbf{All external edges are cycle edges. } Using the same reasoning as above, we get an expected net decrease of at least
    \[\boxed{p\gamma - p \frac{\beta}{3}.} \]
    \item \textbf{The external edges are a combination of non-$K_5$ degree edges, $K_5$ degree edges, and cycle edges.} An averaging argument as in the previous section shows that the three cases above are the only possible worst cases. 
\end{enumerate}

\subsection{Setting the Parameters}
For fixed $\lambda$, finding the optimal values of $\beta, \gamma,
\tau$ can be done with a linear program. In
particular, we maximize the minimum expected decrease $\delta$, which is the minimum of the boxed expressions in \S\ref{sec:cycle-edges} and \S\ref{sec:degree-edges}. The additional constraints are that $\tau \leq \gamma \leq \beta$, $\beta \leq \nf{1}{12}$, $\beta \geq 2\tau$, and $\beta \geq 2\gamma$.  We
then optimize over $\lambda$. This gives $\lambda = 0.71$, $\beta = \nf{1}{12}$, $\gamma = \nf{13}{36}$, $\tau = \nf{13}{36}$, and $\delta = 0.000792$. We have $\varepsilon = 2\delta = 0.001584$, proving \Cref{thm:Ojoinlowering}. 

We note that using the max entropy sampler alone $(\lambda = 1)$ gives $\varepsilon = 0.00146$ and using the matroid intersection sampler alone $(\lambda = 0)$ gives $\varepsilon = 0.00093$.

\section*{Acknowledgments}
The authors thank Nathan Klein and Mehrshad Taziki for pointing out an error in a conditional probability calculation in case (ii) of \Cref{sec:cycle-edges} in a previous version.

{\small
\bibliographystyle{alpha}
\bibliography{bib}
}

\appendix

\section{Details of the Cut Hierarchy Construction}
\label{sec:app-cactus}

\subsection{The Algorithm} \label{welldefinedalgo}

\begin{fact}[Lemma~2 in \cite{FF}] \label{crossingtightsets}
  Let $k \geq 1$ be an integer, $H=(W,F)$ be a $4$-regular, $4$EC graph in which there is a proper min-cut, and every proper min-cut is crossed by some proper min-cut. Then $H$ is a double cycle. 
\end{fact}

\begin{proof}[Proof of \Cref{clm:dcycle}.]
Let $G_f$ refer to the graph in the claim, and $G$ refer to the input to the algorithm.

\emph{Case 1: There is no proper tight set in $G_f$.} Let $u$ be the vertex in $G_f$ such that the two edges that between $r_0$ and $u_0$ in $G$ are between $r_0$ and $u$ in $G_f$. Then either $G_f$ is a graph on 2 vertices with four parallel edges (i.e., a double cycle on 2 vertices), or $\{r_0, u\}$ is a tight set in $G_f$. In the latter case, since $\{r_0, u\}$ is not proper, $G_f$ is a double cycle with 3 vertices.
  \emph{Case 2: There is a proper tight set in $G_f$, and every proper
    tight set is crossed by another proper tight set.} By Fact
  \ref{crossingtightsets}, $G_f$ is a double cycle.
\end{proof}

\begin{proof}[Proof of \Cref{clm:dichotomy}.]
Let $v_{\bar{S}}$ denote the vertex corresponding to the contraction of $V \setminus S$ in $G'$. Suppose $G'$ has a proper min-cut. Let $S'$ be the shore of this min-cut not containing $v_{\bar{S}}$. Then $S' \subset S$. But by minimality of $S$, $S'$ is crossed by another proper tight set. Since this must be true for every proper min-cut $S'$ in $G'$, by Fact \ref{crossingtightsets}, $G'$ must be a double cycle. 
\end{proof}

\subsection{Cactus Representation of Min-Cuts} \label{cactusequiv}

\CutMapping*

To prove the above claim, we recall the cactus representation of min-cuts (see, for e.g., \cite{FF}, for a full exposition). The key idea to prove \Cref{lem:min-cuts-mapping} is that Algorithm \ref{algo} implicitly constructs the cactus representation of min-cuts. 

\begin{definition}
  A \emph{cactus} is a loopless, 2EC graph in which each edge belongs to exactly one cycle.
\end{definition}
Note that this includes cycles on two vertices, containing exactly two parallel edges. The min-cuts of a cactus are obtained by cutting any two edges belonging to the same cycle. 

\begin{theorem}[Cactus Representation of Min-Cuts, \cite{FF}, Theorem 7] \label{cactusthm}
   Let $G$ be a 4-regular, 4EC graph. There is a cactus $C=(U,F)$ and a mapping $\phi: V \rightarrow U$ so that if $U_1$ and $\bar{U_1}$ are two shores of a min-cut in $C$, then $\phi^{-1}(U_1)$ and $\phi^{-1}(\bar{U_1})$ are two shores of a min-cut in $G$. Further, every min-cut in $G$ can be obtained this way for some $U_1$. 
\end{theorem}

\begin{figure}\centering
\includegraphics[scale=.6]{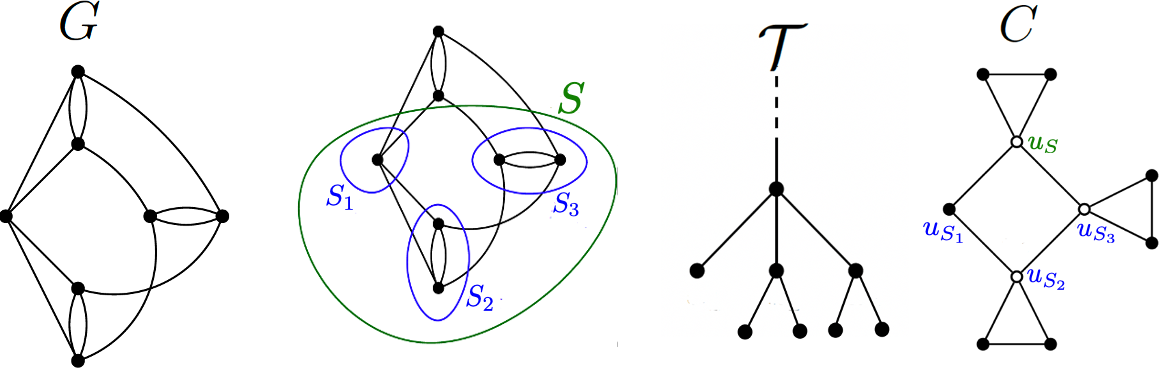}
\caption{Construction of the cactus $C$ from the hierarchy $\cT$ in the proof of \Cref{lem:min-cuts-mapping}.}
\label{fig:cactus}
\end{figure}

\begin{proof}[Proof sketch of \Cref{lem:min-cuts-mapping}.]
  Construct a cactus $C=(U,F)$ as follows. Create a vertex $u_S$ in $U$ for every node $S$ in the hierarchy. If $S$ is a cycle node, make a cycle between $u_S$ and the vertices $u_T$ for all nodes $T$ which are children of $S$. If $S$ is a degree node, include two parallel edges from $u_S$ to each $u_T$ for which $T$ is a child of $S$. See Figure \ref{fig:cactus}. All of the terminal vertices in $U$, except for $u_{V_G \setminus r_0}$, correspond to all of the leaf nodes in $\cT$, which in turn correspond to all of the vertices in $V_G \setminus \{r_0\}$. We identify the vertex $u_{V_G \setminus r_0}$ in $U$ with the vertex $r_0$ in $V_G$. In other words, the mapping $\phi$ in Theorem \ref{cactusthm} can be defined as $\phi(v) = u_v$ for all $v \in V_G \setminus r_0$, and $\phi(r_0) = u_{V_G \setminus r_0}$. 
  
  Since the construction of the cactus in the proof of Theorem \ref{cactusthm} given in \cite{FF} and the construction of $C$ above are equivalent, $\phi$ satisfies the conclusions of Theorem \ref{cactusthm}. Since the min-cuts of a cactus are obtained by cutting any two edges belonging to the same cycle, \Cref{lem:min-cuts-mapping} follows. 
  \end{proof}

\section{Samplers for Odd $|V(H)|$}
\label{sec:odd-sampler}

The main change from the case where $V(H)$ is even is that $x = \nicefrac12 \cdot \ones$
is not a point in the perfect matching polytope: in fact, $H$ has no
perfect matchings because it has an odd number of
vertices.

To fix this, split the external vertex $r$ into two external vertices
$r_1, r_2$, each incident to some two of the edges in $\partial r$,
and then add two parallel edges between $r_1, r_2$. Call this multigraph
$\Hh$. Again, all vertices and edges between vertices in
$V \setminus \{r_1, r_2\}$ are called internal, vertices in
$\{r_1, r_2\}$ and edges in $\partial r_1 \cup \partial r_2$ are
external, and internal vertices adjacent to external vertices are 
boundary vertices.

\begin{fact}
  The graph $\Hh$ is also $4$-regular and 4EC; it has an even
  number of vertices. However, it contains a single proper
  min-cut separating the internal vertices
  $I := V \setminus \{r_1, r_2\}$ from the two external vertices.  Any
  perfect matching $M$ of $\Hh$ contains either zero or two
  edges in $\partial I$.
\end{fact}

The sampling procedure for trees on $H[I] = \Hh[I]$ is conceptually similar to
the one above, but with some crucial changes because of the
proper min-cut.
\begin{enumerate}
\item Since $\Hh$ is $4$-regular and 4EC, and
  $|V(\Hh)|$ is even, setting $\nf14 = \nf{x_e}{2}$ on
  each edge gives a solution $K_{PM}(\Hh)$ in the perfect
  matching polytope. Sample a perfect matching $M$ such that for each
  edge $e$, we have $\Pr[e \in M] = \nf14$.  Again define the fractional
  solution
  \[ y_e = \nf13 + \ones_{(e \in M)} \cdot \nf23. \] This ensures
  $\E[y] = x$, but $y|_I$ may no longer be in the spanning tree polytope
  $\Kst(\Hh[I])$. Indeed, $y(\partial I) \in \{\nf43, \nf83\}$ depending on
  whether $M$ has zero or two edges crossing the cut, so $y(E(I)) \in
  \{|I|-\nf23,|I|-\nf43\}$ and not $|I|-1$ as required. Hence we cannot directly sample a spanning tree on
  $\Hh[I]$, as we did above. We fix this in
  step~\ref{item:fix-feas} below.

\item Pick a random induced sub-matching $M'$ of $M$, and define
  partition matroid constraints on the sets $L_{uv}, R_{uv}$ incident on the endpoints of %
  edges $uv$ in this sub-matching $M'$ exactly as in
  steps~\ref{item:pick-Mp-even} and~\ref{item:beetles-even} of the sampling
  procedure for the even-size case above. 

\item \label{item:fix-feas} We now address that $y|_I$ is not in the
  spanning tree polytope %
  by locally changing the fractional solution
  as follows:
  \begin{enumerate}
  \item (Local Decrease) If $M \cap \partial I = \emptyset$, then
    $y(E(I)) = |I|-\nf23$ and $y(\partial I) = \nf43$. Hence $y|_I$ is
    not a solution to $\Kst(\Hh[I])$.  To fix this, pick a random edge
    $e \in \partial I$, pick a random edge $f \in E(I)$ adjacent to
    $e$, and reduce the $y$-value of $f$ by $\nf13$ to get solution
    $\yh$. (Note this may reduce either a $1$-edge to $\nf23$, or a
    $\nf13$-edge to zero.)

  \item (Local Increase) If $|M \cap \partial I| = 2$, then
    $y(E(I)) = |I|-\nf43$, and hence $y|_I \not\in \Kst(\Hh[I])$.
    To fix this, pick a random edge $e$ from the two matching edges in
    $M \cap \partial I$, pick a random edge $f \in E(I)$ adjacent to
    $e$, and increase the $y$-value of $f$ by $\nf13$ to get solution
    $\yh$. (Note this increases one $\nf13$ edge to $\nf23$.)

    Suppose this edge $f$ is contained in some matroid constraint, say
    on set $P$ of edges:
    \begin{enumerate}
    \item If all three edges $f,g,h$ of this constraint are in
      $E(I)$, pick a random one of $\{g,h\}$, say $h$, and redefine
      $P \gets \{f,g\}$. Note that $\yh_f + \yh_g = 1$ and the
      constraint is tight.
    \item If $f$ belongs to a constraint $P = \{f,g\}$, that used to
      be $\{f,g,h\}$ but the third edge $h$ was in $\partial I$ and
      hence was dropped, retain the constraint $P = \{f,g\}$. Again
      $\yh_f + \yh_g = 1$.
    \end{enumerate}
  \end{enumerate}
\end{enumerate}

The following claims show that $\yh|_I$ has the right marginals, and
that it belongs to the spanning tree polytope. So now we can sample
a spanning tree $T$ from it (subject to the partition matroid $\cM$).
Then we show, in \S\ref{sec:corr-prop-odd}, that this tree also
satisfies negative correlation properties.%

\subsection{Feasibility of the Solution \texorpdfstring{$\yh$}{yhat}}
\label{sec:feasible}

\begin{claim}[Marginal Preserving]
  $\E[\yh_e] = x_e = \nf12$ for all $e \in E(I)$.
\end{claim}
\begin{proof}
  Since $\E[y_e] = \nf12$ for all $e$, we focus on the expected
  difference $\E[\yh_e - y_e]$. We only need to worry about internal
  edges incident to boundary vertices, since those are the only edges
  whose $y$-values are changed.

  \begin{itemize}
  \item Let $e$ be an edge incident only on one boundary vertex $v$,
    that has a boundary edge $f$. Now consider choosing $M$ randomly.
    When $M \cap \partial I = \emptyset$, which happens \wp $\nf12$,
    $e$ is decreased \wp $\nf1{12}$, since the total internal degree
    of boundary vertices is $12$. On the other hand, when $f \in M$ (\wp
    $\nf14$, and this implies the event $|M \cap \partial I| = 2$),
    $e$ is increased with probability $\nf16$, since the total
    internal degree of the boundary vertices incident to internal
    endpoints of $M$ is $6$. 

  \item Let $e = uv$ such that both $u,v$ are boundary vertices with
    edges $f,g \in \partial I$ respectively. Now, when 
    $M \cap \partial I = \emptyset$, which happens \wp $\nf12$, $e$ is
    decreased \wp $\nf16$ (since $e$ is counted from both ends).  When
    $f \in M$ (\wp $\nf14$, and this implies the event
    $|M \cap \partial I| = 2$), $e$ is increased with probability
    $\nf16$ ``through $f$''. And same happens for $g$.  Therefore,
    regardless of how $f$ and $g$ are correlated, $e$ is increased
    with probability $2 \cdot \nf14 \cdot \nf16 = \nf1{12}$. 
  \end{itemize}
  In both cases, the expected difference is zero, giving the proof.
\end{proof}

\begin{claim}[Feasible]
  The solutions $\yh|_{I}$ belong to the spanning tree polytope.
\end{claim}
\begin{proof} 
  The equality constraint $\sum_{e \in I} \yh_e = |I|-1$ is satisfied
  by design, so we need to check that
  $\sum_{e \in E[S']} \yh_e \leq |S'| - 1$ for all proper subsets
  $S' \subsetneq I$. Recall that $\partial_H S'$ has at least six
  edges, including the edges leaving $I$.  Let
  $a := |E_H(S', I \setminus S')|$ and $b := |E_H(S', V\setminus I)|$ so
  that $a + b \geq 6$.

  Observe that
  \begin{gather}
    \sum_{e \in E[S']} y_e = \frac12 \bigg( \sum_{v \in S'} \sum_{e
      \in \partial_H v} y_e - \sum_{e \in \partial_H S'} y_e \bigg)
    \leq \frac12 \Big( 2|S'| - (a+b)/3 \Big) \leq |S'| - 1. \label{eq:spanning_tree}
  \end{gather}
  \begin{itemize}
  \item In a local decrease, since we decrease $y$ to get $\yh$, the constraint is still satisfied.
  \item In a local increase, since we increase only one edge by $\nf13$, we
    only need to consider the case when~\eqref{eq:spanning_tree} is
    tight, which implies that every edge going out of $S'$ has value
    exactly $\nf13$.  In particular, every vertex in $S'$ is matched inside $S'$.
    By construction, the edge whose
    $y$-value we increase must have a neighboring edge matched by $M$ outside $I$.
    This $M$-edge cannot be incident to $S'$, so we are safe. \qedhere
  \end{itemize}
\end{proof}

\subsection{Correlation Properties: $|V(H)|$ odd}
\label{sec:corr-prop-odd}

To prove the correlation properties in the odd case, we observe that each of the proofs
from~\S\ref{sec:corr-prop-even} focuses on a small set of
edges. Hence, if we can show that with constant probability the
fractional solution for these edges is unchanged, the proofs go through 
with small changes (and  weaker constants). A more
careful case-analysis is likely to  give better constants.

We will often make arguments that condition on certain sets of edges not being perturbed by a local increase or decrease. The following lemma gives lower bounds on the probabilities that these desired events occur.

\begin{lemma}\label{lemma:untouched}
Let $f, g$ be internal edges incident to vertex $v$. Then, if $f$ or $g$ lies in $M$, $y_f$ and $y_g$ remain untouched by a local increase or decrease \wp at least $\nf23$. Similarly, if $v$ is an internal vertex with internal edges $e, f, g, h$, the four edges remain untouched with probability at least $\nf23$
\end{lemma}
\begin{proof}
  If one of $f,g$ lies in $M$ (which happens \wp \nf12), then its $y$
  value equals $1$ and it belongs to $T$ \wp $1$. If
  there is a local increase, then the increase can only affect the
  edge not in $M$ (through its other endpoint), so the $y$-value
  remains unchanged \wp $\nf23$. If there is a local
  decrease, both edges $f,g$ may have both endpoints that are boundary
  vertices in worst case, giving a $\nf4{12}$ chance of being decreased,
  and hence a $\nf8{12} = \nf23$ chance of
  $y$-values of these two edges remaining unchanged.
  
  The second statement follows from the fact that there are three internal edges adjacent to an external edge. 
\end{proof}

We move on to the proofs of the correlation properties. We restate the lemmas here for convenience. 
  
\Oddfg*

\begin{proof}[Proof of \Cref{clm:odd1}, Odd Case]
  \emph{The \MI claims:} As in the even case, to prove $\pr(|T \cap \{f,g\}|=2) \geq \nf{1}{9}$, we need only knowledge of the marginals and not the specific sampler. If one of $f,g$ lies in $M$ (which happens \wp \nf12)
  $y$-values of these two edges remaining unchanged with probability at least \nf23, due to Lemma \ref{lemma:untouched}. Now, conditional on these $y$-values remaining unchanged, the other edge (not in $M$) is chosen with probability $\nf13$, making the
  unconditional probability $\nf12 \cdot \nf23 \cdot \nf13 = \nf19$,
  as desired.  Similarly, conditioned on $f$ lying
  in $M$ and hence belonging to $T$, edge $g$ is not chosen \wp
  $1 - y_e = \nf23$, giving a 
  probability at least $\nf14 \cdot \nf23 \cdot \nf23 = \nf19$, proving the second part of \Cref{clm:odd1} for the \MI sampler.

  \emph{The \ME claim:} 
  It remains to show that $\pr(T \cap \{f,g\} = \{f\}) \geq \nf{12}{72}$. In the case that $f \in M$, we once again have a probability of at least $\nicefrac{2}{3}$ that $f$ and $g$ are untouched, due to Lemma \ref{lemma:untouched}. Sincce $\pr(g \not \in T) = \nf{2}{3}$ in this case, we have 
\[ \pr(T\cap\{f,g\}=\{f\} \wedge f \in M) \geq \nicefrac{1}{4} \cdot \nicefrac{2}{3} \cdot \nicefrac{2}{3} = \nicefrac{1}{9}. \]
In the case that neither $f$ nor $g$ is in $M$, the probability $f$ or $g$ is chosen in a local decrease is at most $\nicefrac{1}{3}$ and the probability $f$ or $g$ is chosen in a local increase is at most $\nicefrac{1}{2}$. (For the first bound, the worst case is when each of the three endpoints of $f,g$ is incident to an external edge. Then the probability that neither $f$ nor $g$ is chosen for a local decrease is at least $\nf12 \cdot \nf23 + \nf14 \cdot \nf13 + \nf14 = \nf23$. For the second bound, the worst case is when each edge in $M \cap \partial I$ is incident to an endpoint of $f$, w.l.o.g. Then the probability that neither $f$ nor $g$ is chosen for a local increase is at least $\nf12 \cdot \nf23 + \nf12 \cdot \nf13 = \nf12$.)  Hence, the edges remain untouched with probability at least $\nicefrac{1}{2}$. In total, this gives
\[ \pr(T \cap \{f,g\} = \{f\} \wedge f,g \not \in M) \geq \nicefrac{1}{2} \cdot \nicefrac{1}{2} \cdot \nicefrac{2}{9} = \nicefrac{4}{72} \]
where the bound $\pr(T \cap \{f,g\} = \{f\} \mid f,g \not \in M) \geq \nf19$ was computed in the even case (\ref{3.1evencase}). 
Hence, 
\[ \pr(T \cap \{f, g\} = \{f\}) \geq \nicefrac{1}{9} + \nicefrac{4}{72} = \nicefrac{12}{72}. \qedhere \]
\end{proof}

\Oddefgh*

\begin{proof}[Proof of \Cref{clm:odd2}, Odd Case]
  \emph{The \MI claims:} Each perfect matching $M$ contains one of these four edges in
  $\partial v$. Say that edge is $e$.
  The vertex $v$ has no external edges, so
  regardless of whether we do a local increase or decrease, these
  four edges remain untouched \wp
  $\nf23$, due to Lemma \ref{lemma:untouched}. If this happens, and if $e$ also belongs to $M'$ (\wp
  $\nf17$), tree $T$ contains exactly one of $\{f,g,h\}$, which gives
  us the bound $\nf{2}{21}$. Moreover,
  the probability of this edge in $T$
  belonging to the other pair (in this case, $\{g,h\}$) is $\nf23$,
  giving the bound $\nf{4}{63}$.
  
  \emph{The \ME claims:} Since $v$ is not a boundary vertex, none of $e, f, g, h$ are picked for a local increase or decrease \wp at least $\nf23$, due to Lemma \ref{lemma:untouched}. Conditioning on $e, f, g, h$ remaining untouched, we have the same analysis as that of the even case (\ref{3.2evencase}). Hence, we obtain 
\[ \pr(|T \cap \{e, f, g, h\}| = 2) \geq \nicefrac{2}{3} \cdot \nicefrac{4}{9} = \nicefrac{8}{27} \]
and 
\[ \pr(|T \cap \{e, f\}| = |T \cap \{g, h\}| = 1) \geq \nicefrac{2}{3} \cdot \nicefrac{8}{27} = \nicefrac{16}{81}. \qedhere \]
\end{proof}

\OddGen*

\begin{proof}[Proof of \Cref{clm:odd4}a, Odd Case]
  \emph{The \MI claims:} In this proof, we will often use that $G$ is a simple graph.
    Suppose $e \in M'$, which happens \wp $\nf1{28}$. Call the edges
    $(\partial u \cup \partial v) \setminus \{e\}$ \emph{interesting}
    edges.  Condition on the matching $M$ (which decides whether the
    change is a local increase or a decrease), and the edge $f$ in
    $\partial I$ whose incident internal edge is increased/decreased.
    At most two of the three internal edges
    incident to $f$ can be interesting. Since $u,v$ are non-boundary
    vertices, the value of $e$ is never changed by the alteration.

    First suppose it is a local decrease, then with probability at
    least $\nf13$ a non-interesting edges is reduced, and hence the
    matroid constraint ensures we always succeed. With probability at
    most $\nf23$ an interesting edge is reduced, in which case we
    still have a $\nf23$ chance of getting degree $2$ at both $u$ and
    $v$. So we have $\nf13 + \nf23\cdot\nf23 = \nf79$ of success,
    conditioned on the choice of the edge.

    Else, suppose it is a local increase; the argument is
    similar. With probability at least $\nf13$ a non-interesting edges
    is increased, and hence the matroid constraint ensures we always
    succeed. With probability at most $\nf23$ an interesting edge is
    increased, in which case we change the matroid constraint to drop
    one of the other interesting edges. This dropped edge can be added
    in with its $\yh$-value $\nf13$, and hence we have a $\nf23$
    probability of getting degree $2$ at both $u$ and $v$. Again, the
    probability of success is $\nf13 + \nf23\cdot\nf23 = \nf79$ of
    success.
    Removing the conditioning on $e \in M'$, the net probability of
    success is $\nf{1}{28} \cdot \nf79 = \nf1{36}$.
    
    \emph{The \ME claims: }Let $U = \delta(u) \setminus \{e\}$ and $V = \delta(v) \setminus \{e\}$. Condition on $e \in M$. Since $u,v$ are both non-boundary vertices, the value of $e$ will remain unchanged regardless of whether there is a local increase or decrease, so we will have $y_e = 1$ always. Recall that $y$ is changed at a single edge when the local increase or decrease is performed. We take cases on whether a local increase or local decrease occurs.

\begin{enumerate}[wide, labelwidth=!, itemsep=1.5ex]

\item[\underline{A local decrease occurs.}] 

\begin{figure}[h]
    \centering
    \includegraphics[scale=0.5]{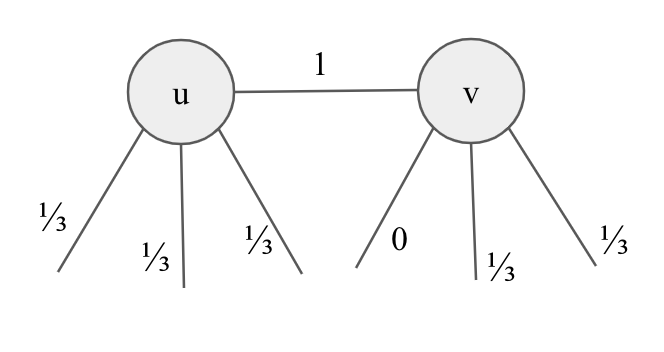}
    \caption{The decrease case in \Cref{clm:odd4}a, labelled with $y$-values after the decrease is performed.}
    \label{fig:6a-decrease}
\end{figure}
If none of the edges in $U \cup V$ are decreased, then (conditioned on $e \in M$) the analysis from the even case (\ref{3.3Aevencase}) gives 
\[\pr(X_{U} = 1 \wedge X_{V} = 1) \geq \nf{256}{2187}.\] 
Now consider the case that an edge in $U \cup V$ is decreased. Let $S_1 = U$ and let $S_2 = \{a, b\}$ denote the set of two nonzero edges in $V \setminus \{e\}$ (see \Cref{fig:6a-decrease}). Applying \Cref{thm:Hoeff} and the fact that $\ex(|S_2 \cap T|) = \nicefrac{2}{3}$, we have $\pr(|S_2 \cap T| = 1) \geq 2 \cdot \nicefrac{2}{3} \cdot \nicefrac{1}{3} = \nicefrac{4}{9}$.
By \Cref{thm:SRfacts}, 
\[ \nicefrac{1}{3} = 1 - \nicefrac{2}{3} \leq \ex(|S_1 \cap T| \mid a \in T, b \not \in T) \leq 1 +  \nicefrac{1}{3} = \nicefrac{4}{3}. \]
We apply \Cref{thm:Hoeff} again to lower bound $\pr(|S_1 \cap T| = 1 \mid a \in T, b \not \in T)$ by 
\[ \pr(|S_1 \cap T| = 1 \mid a \in T, b \not \in T) \geq 3 \cdot \nicefrac{1}{9} \cdot (\nicefrac{8}{9})^2 = \nicefrac{64}{243}. \]
Since the same calculation may be repeated with $b \in T$ and $a \not \in T$ (as $y_a = y_b$), we have 
\[ \pr(X_U = 1 \wedge X_V = 1) \geq \nicefrac{4}{9} \cdot \nicefrac{64}{243} = \nicefrac{256}{2187}. \]
So comparing both cases, we have that conditioned on $e \in M$, if there is a local decrease then \[\pr(X_U=1 \wedge X_V=1) \geq \min\{\nf{256}{2187}, \nf{256}{2187}\}=\nf{256}{2187}.\]

\item[\underline{A local increase occurs.}] We first note that the probability of no edge in $U \cup V$ being increased is at least $\nf{1}{3}$, in which case the analysis from the even case (\ref{3.3Aevencase}) gives \[\pr(X_{U} = 1 \wedge X_{V} = 1) \geq \nf{256}{2187}.\]
Now consider the case that an edge $c$, say in $U$, is increased from $\nf13$ to $\nf23$. Let $U \setminus \{c\} = \{a,b\}$. Since $\ex[X_{a,b}] = \nf{2}{3}$, by Markov's inequality we have $\pr(X_{a,b} = 0) \geq \nicefrac{1}{3}$. Furthermore, $ \ex[X_c \mid X_{a, b}=0] \geq \ex[X_c] = \nf{2}{3}$ by \Cref{thm:SRfacts}. In total, $\pr(c \in T, a, b \not \in T) \geq \nf{1}{3} \cdot \nf{2}{3} =\nf{2}{9}$.
Applying  \Cref{thm:SRfacts} thrice more, we obtain
\[ \nicefrac{1}{3} \leq \ex(|S_2 \cap T| \mid c \in T, a, b \not \in T) \leq \nicefrac{5}{3} \]
which by \Cref{thm:Hoeff} gives us 
\[ \pr(|S_2 \cap T| = 1 \mid c \in T, a, b \not \in T) \geq \nicefrac{64}{243}. \]
This yields, in the case that an edge in $U \cup V$ is increased, 
\[ \pr(X_U = 1 \land X_V = 1) \geq \nicefrac{64}{243} \cdot \nicefrac{2}{9} = \nf{128}{2187}. \]
Since the probability that no edge in $U \cup V$ is increased is at least $\nf13$, and $\nf{256}{2187} > \nf{128}{2187}$, in the worst case we have that, conditioned on $e \in M$, if there is a local increase then
\[ \pr(X_U = 1 \land X_V = 1) \geq \nicefrac{1}{3} \cdot \nicefrac{256}{2187} + \nicefrac{2}{3} \cdot \nicefrac{128}{2187} = \nicefrac{512}{6561}. \]
\end{enumerate}
Now comparing the cases of local increase and decrease and removing the conditioning of $e \in M$, we obtain in the worst case 
\[\pr(X_U = 1 \land X_V = 1) \geq \nicefrac{1}{4} \cdot \min\{\nf{256}{2187}, \nf{512}{6561}\} = \nf{128}{6561} \]
which is the bound we sought. 
\end{proof}

\begin{proof}[Proof of \Cref{clm:odd4}b, \emph{Odd Case}.]
    \emph{The \MI claims.} Fix a matching $M$ such that $e \in M$
    (\wp $\nf14$); each of $u,v$ have two other internal edges, each
    with $y$-value $\nf13$. Call these four \emph{interesting} edges.
    Suppose $M \cap \partial I$ has two edges, then we perform a local
    increase step. Condition on the edge in $M \cap \partial I$ chosen
    as part of the local increase: in the worst case it is incident to
    the endpoint of two of the interesting edges. 
    \begin{itemize}
    \item Then \wp $\nf23$ it increases one of the interesting edges
      from $y_f = \nf13$ to $\yh_f = \nf23$. The total $\yh$-value of
      interesting edges becomes $\nf53$. By the same argument as for
      the even case, the tree $T$ contains exactly one
      interesting edge with probability at least $\nf13$.
    \item With the remaining probability $\nf13$, the increase is not
      to an interesting edge, and hence we succeed with the original
      probability $\nf23$.
    \end{itemize}
    The total probability of success is therefore
    $\nf23 \cdot \nf13 + \nf13 \cdot \nf23 = \nf49$.
        
    On the other hand, if $M \cap \partial I$ is empty, then we
    perform a local decrease step. Condition on the boundary edge
    chosen in the first part of the local decrease. In the worst case
    it is one of the boundary edges incident to either $u$ or $v$.
      One of the interesting edges is reduced \wp $\nf23$, from
      $y_f = \nf13$ to $\yh_f = 0$. If that happens, the total
      $\yh$-value of interesting edges sums to $1$.  The connectivity
      of the spanning tree means that every tree we sample with these
      marginals will have exactly one interesting edge. 

    Hence, in either case, we have exactly one of these four internal
    edges and satisfy the condition of the claim with probability
    $\nf49$, conditioned on having $e \in M$. The overall probability
    is therefore at least $\nf14 \cdot \nf49 = \nf1{9}$.
    
    \emph{The \ME claims.} Observe that the edge $e$ will always contribute either: 0 to both the degree of $u$ in $T$ and the degree of $v$ in $T$ OR 1 to both the degree of $u$ in $T$ and the degree of $v$ in $T$. Thus, to lower bound the probability that $u$ and $v$ have different parities in $T$, we will not need to consider whether or not $e$ is in $T$. 

Let $a,b$ be the two \emph{internal} edges in $\delta(u) \setminus \{e\}$ and $c,d$ be the two \emph{internal} edges in $\delta(v) \setminus \{e\}$. (We know each set has exactly two edges because $u,v$ are each boundary vertices.) Let $f,g$ be the two external edges incident to $u,v$, respectively. We take cases based on (a)~whether a local increase or decrease is performed, (b)~whether or not any of $a,b,c,d$ are changed by the local increase or decrease, and (c)~whether or not $e \in M$. Note that for (b), we do not take subcases on whether or not $e$ is changed, by the above paragraph. Also note that for (c), we consider whether $e \in M$ in order to determine the marginals of $\{a,b,c,d\}$ prior to a local increase or decrease.

\begin{enumerate}[wide, labelwidth=!, itemsep=1ex]
\item[\underline{Case I.N.1: local increase occurs, $\{a, b, c, d\}$ not touched, $e \in M$.}] In this case, we obtain the same bound from Case 1 of the even case in the proof of \Cref{clm:odd4}b (see \ref{3.3bevencase}), that is,
\[ \pr(\text{exactly one of $u,v$ has odd degree in $T$}) \geq \pr(X_{a,b,c,d}=1) \geq \boxed{\nicefrac{2}{3}}. \]

\item[\underline{Case I.N.2: local increase occurs, $\{a, b, c, d\}$ not touched, $e \not \in M$.}] This case includes the two subcases $a,c \in M$ and $a \in M, c \not \in M$, which are cases 2 and 3 in the proof of \Cref{clm:odd4}b (see \ref{3.3bevencase}) and yield lower bounds on the desired probabilitiy of $\boxed{\nf49}$ and $\boxed{\nf{8}{27}}$, respectively. Finally, there is one more subcase for the case where $|V(H)|$ is odd, namely, $a,c \not \in T$, since we may have $f,g \in M$. Since in this subcase $a,b,c,d$ are all $\nf13$-valued edges, we have the same bound from Case 1 of the even case in the proof of \Cref{clm:odd4}, namely, $\boxed{\nf23}$. 

\item[\underline{Case I.T.1: local increase occurs, $\{a, b, c, d\}$ touched, $e \in M$.}] One of the edges $a, b, c, d$ is increased from $\nicefrac{1}{3}$ to $\nicefrac{2}{3}$, while the rest are still $\nf13$-valued, since $e \in M$. Hence, $\ex[X_{a,b,c,d}] = \nicefrac{5}{3}$.
Furthermore, note that we must have $X_{a,b,c,d} \geq 1$. This gives  
\[ \pr(\text{exactly one of $u,v$ has odd degree in $T$}) \geq \pr(X_{a, b, c, d} = 1) \geq \boxed{\nf{1}{3}}. \]

\item[\underline{Case I.T.2: local increase occurs, $\{a, b, c, d\}$ touched, $e \not \in M$.}] Just as in Case I.N.2., there are three subcases. First consider the subcase of $a,c \in M$. In this subcase, $a,c$ cannot be increased, so we know $\ex[X_{b,d}] = \nf13$. By \Cref{thm:Hoeff}, $\pr(X_{b,d}=1) \geq \boxed{\nf12}$. 

Next, consider the subcase that only $a \in M$. Then, we have $\ex(X_{b, c, d})=\nf{4}{3}$
and hence by \Cref{thm:Hoeff}, 
$\pr(X_{b, c, d} = 2) \geq 2 \cdot \nf13 \cdot \nf13 \geq \boxed{\nf{5}{18}}$.

For the final subcase that $f,g \in M$, we have $\ex[X_{a, b, c, d}] = \nf{5}{3}$
and since $X_{a,b,c,d} \geq 1$, we have $\pr(X_{a, b, c, d} = 1) \geq \boxed{\nf{1}{3}}$. 
\end{enumerate}

This concludes the local increase cases. We move on to the local decrease cases.

\begin{enumerate}[wide, labelwidth=!, itemsep=1ex]
\item[\underline{Case D.N.1: local decrease occurs, $\{a, b, c, d\}$ not touched, $e \in M$.}] This is the same case as I.N.1., giving a bound of $\boxed{\nf23}$.

\item[\underline{Case D.N.2: local decrease occurs, $\{a, b, c, d\}$ not touched, $e \not \in M$.}] This is the same case as I.N.2., taking only the first subcase ($a,c \in M$), since $M \cap \partial I = \emptyset$ in a local decrease. Thus the lower bound here is $\boxed{\nf49}$.

\item[\underline{Case D.T.1: local decrease occurs, $\{a, b, c, d\}$ touched, $e \in M$.}] Since $e \in M$, we know $a, b, c, d$ start of as $\nicefrac{1}{3}$-valued. Then, w.l.o.g., say $a$ is decreased to 0. Then, $\ex[X_{b, c, d}] = 1.$ But since $X_{b, c, d} \geq 1$, this gives $\pr(X_{b, c, d} = 1) = \boxed{1}.$

\item[\underline{Case D.T.2: local decrease occurs, $\{a, b, c, d\}$ touched, $e \not \in M$.}] Since in a local decrease $M \cap \partial I = \emptyset$, we have w.l.o.g. that $a,c \in M$ and that $a$ is not decreased. So $\ex(X_{b, c, d}) = 1 + 2 \cdot \nf13 - \nf13 = \nf{4}{3}.$ Since we know $a \in T$, we instead study $\pr(X_{b, c, d} = 2)$; by \Cref{thm:Hoeff}, we have 
$\pr(X_{b, c, d} = 2) \geq 2 \cdot \nf16 \cdot \nf56 = \boxed{\nf{5}{18}}$.
\end{enumerate}
The minimum of all the boxed quantities is $\nf{5}{18}$. Hence, for the odd case of \Cref{clm:odd4}b, we obtain the desired bound
\[ \pr(\text{exactly one of $u,v$ has odd degree in $T$}) \geq \nf{5}{18}. \qedhere \]
\end{proof}

\section{Charging Proof Details} \label{sec:charging proofs}

\vspace{0.3cm}

\begin{proof}[Proof of Averaging Argument, \S\ref{sec:nonsp}, Case 4]
We will use an averaging argument to show that the expected charge is no worse than $p \cdot \max\{\frac{\tau}{2}, \frac{\gamma}{2}, \frac{\beta}{4}\}.$ We construct a flow network as in the proof of \Cref{hall} in order to show there exists a function $x: B \times F \rightarrow \mathbb{R}_+$ satisfying constraint (\ref{upper bound}) for $c = \max\{\frac{\tau}{2}, \frac{\gamma}{2}, \frac{\beta}{4}\}$. Then we can apply a charging scheme like the one described in Case 1 using the function $x$. The flow network is as follows: take the graph $H=(B,F)$ and add a source node $s$ and a sink node $t$. Connect $s$ to every vertex $b$ in $B$, and label $(s,b)$ with capacity $\tau, \gamma,$ or $\frac{\beta}{2}$ based on whether $b$ corresponds to a non-$K_5$ degree, $K_5$ degree, or cycle edge, respectively. To find $x$ satisfying the properties in the above paragraph, by Lemma 1 we need to show that for each $T \subseteq B$, 
\[|N(T)| \geq \frac{\tau \cdot n_{\tau} + \gamma \cdot n_{\gamma} + \frac{1}{2} \cdot \beta \cdot n_{\beta}}{\max\{\frac{\tau}{2}, \frac{\gamma}{2}, \frac{\beta}{4}\}}\]
where $n_{\tau}, n_{\gamma},$ and $n_{\beta}$ are the numbers of vertices in $T$ corresponding to non-$K_5$ degree, $K_5$ degree, and cycle edges, respectively. Note $n_{\tau} + n_{\gamma} + n_{\beta} = |T|$. 

 \Cref{Hall's constants} showed that 
\begin{align*}
|N(T)| &\geq \frac{\tau \cdot |T|}{\frac{\tau}{2}} \geq \frac{\tau \cdot |T|}{\max\{\frac{\tau}{2}, \frac{\gamma}{2}, \frac{\beta}{4}\}} \\
|N(T)| &\geq \frac{\gamma \cdot |T|}{\frac{\gamma}{2}} \geq \frac{\gamma \cdot |T|}{\max\{\frac{\tau}{2}, \frac{\gamma}{2}, \frac{\beta}{4}\}} \\
|N(T)| &\geq \frac{\frac{1}{2} \cdot \beta \cdot |T|}{\frac{\beta}{4}} \geq \frac{\frac{1}{2} \cdot \beta \cdot |T|}{\max\{\frac{\tau}{2}, \frac{\gamma}{2}, \frac{\beta}{4}\}}.
\end{align*}
Now taking a convex combination of the three inequalities above we obtain:
\begin{align*}
|N(T)| &= \frac{n_{\tau}}{|T|} \cdot |N(T)| + \frac{n_{\gamma}}{|T|} \cdot |N(T)| + \frac{n_{\beta}}{|T|} \cdot |N(T)| \\
&\geq \frac{ \tau \cdot n_{\tau} + \gamma \cdot n_{\gamma} +  \frac{1}{2} \cdot \beta \cdot n_{\beta}}{\max\{\frac{\tau}{2}, \frac{\gamma}{2}, \frac{\beta}{4}\}}
\end{align*}
which is precisely what we sought to show. This concludes Case 4. 
\end{proof}

\section{Symmetry Lemma}
\label{sec:other-lemmas}

\begin{claim}[Symmetry Lemma]
  \label{claim:symmetry} For any graph $H$ with each edge $e$ having $x_e = \nf12$, let $T$ be
a random set of edges faithful to the marginals, i.e.,
$\Pr[e \in T] = x_e$.
  For any two edges $f,g$, let $p_{00}, p_{10}, p_{01}, p_{11}$ denote
  the probabilities of $T \cap \{f,g\}$ being
  $\emptyset, \{f\}, \{g\}, \{f,g\}$ respectively. Then
  $p_{00} = p_{11}$ and $p_{01} = p_{10}$.
\end{claim}

\begin{proof}
  Observe that $p_{00} + p_{01} + p_{10} + p_{11} =1$. Moreover,
  $p_{01} + p_{10} + 2p_{11} = \E[|T \cap \{f,g\}|] = x_f + x_g =
    1$. 
  Hence, $p_{00} = p_{11}$, proving the first statement. Now,
  $ p_{10} + p_{11} = x_f = x_g = p_{01} + p_{11},$
  giving $p_{10} = p_{01}$, proving the second.
\end{proof}

\end{document}